\documentclass[11pt]{article} 

\usepackage{fullpage}

\usepackage{epsfig} 
\usepackage{amsmath,amsthm,paralist,bm,yhmath}
\usepackage{amssymb}
\usepackage{epstopdf}
\usepackage{graphicx}

\usepackage{cite}

\usepackage[footnotesize]{caption}

\makeatletter
\let\NAT@parse\undefined
\makeatother
\usepackage[sort&compress, numbers]{natbib}

\usepackage{upgreek}
\usepackage[usenames]{color}
\newtheorem{theorem}{Theorem}

\newtheorem{lemma}{Lemma}

\newcommand{\sign}{{\rm sign}\,}

\newcommand{\cl}{\mathcal}
\newcommand{\bs}{\boldsymbol}

\newlength{\imgwidth}
\usepackage{hyperref}
\hypersetup{colorlinks=true, linkcolor=black,  anchorcolor=black,
citecolor=black, filecolor=black, menucolor=black, pagecolor=black, urlcolor=black} 

\title{\textbf{Regime Change: Bit-Depth versus\\ Measurement-Rate in Compressive Sensing}\thanks{This work was supported by the
    grants NSF CCF-0431150, CCF-0728867, CCF-0926127, CNS-0435425, and
    CNS-0520280, DARPA/ONR N66001-08-1-2065, N66001-11-1-4090, ONR N00014-07-1-0936,
    N00014-08-1-1067, N00014-08-1-1112, and N00014-08-1-1066, AFOSR
    FA9550-07-1-0301 and FA9550-09-1-0432, ARO MURI W911NF-07-1-0185
    and W911NF-09-1-0383, and the Texas Instruments Leadership
    University Program.}
\author{Jason N. Laska and Richard G. Baraniuk\thanks{Department of Electrical and Computer Engineering,
    Rice University, Houston, TX, 70015 USA. Email: laska@rice.edu,
    richb@rice.edu.}}} 
    \date{\today}

\begin{document}
\maketitle
\begin{abstract} 
\noindent The recently introduced \emph{compressive sensing} (CS) framework enables digital signal acquisition systems to take advantage of signal structures beyond bandlimitedness.  Indeed, the number of CS measurements required for stable reconstruction is closer to the order of the signal complexity than the Nyquist rate.  To date, the CS theory has focused on real-valued measurements, but in practice, measurements are mapped to bits from a finite alphabet.  Moreover, in many potential applications the total number of measurement bits is constrained, which suggests a tradeoff between the number of measurements and the number of bits per measurement.  We study this situation in this paper and show that there exist two distinct regimes of operation that correspond to high/low signal-to-noise ratio (SNR).  In the \emph{measurement compression} (MC) regime, a high SNR favors acquiring fewer measurements with more bits per measurement; in the  \emph{quantization compression} (QC) regime, a low SNR favors acquiring more measurements with fewer bits per measurement.  A surprise from our analysis and experiments is that in many practical applications it is better to operate in the QC regime, even acquiring as few as 1 bit per measurement.
\end{abstract}

\section{Introduction}
\label{sec:intro}

The \emph{compressive sensing} (CS) framework has sparked renewed interest in sampling and signal acquisition~\cite{Can::2006::Compressive-sampling,Don::2006::Compressed-sensing}. The framework can be concisely summarized by three fundamental components: \emph{i}) \emph{underdetermined linear measurement systems}, i.e., we obtain the measurements
\begin{equation}
\label{eq:ypx}
\bs y = \Phi \bs x + \bs e,
\end{equation}
of the signal  $\bs x \in \mathbb{R}^{N}$, with $\Phi \in \mathbb{R}^{M\times N}$ and $M\ll N$, and with measurement error $\bs e \in \mathbb{R}^{M}$;  \emph{ii}) \emph{structured signal models}, such as $K$-sparse signals, i.e., $\bs x \in \Sigma_{K} := \{\bs x \in \mathbb{R}^{N}: \|\bs x\|_{0} :=
|\mathrm{supp}(\bs x)| \leq K\}$;  and \emph{iii}) \emph{computational reconstruction}, one example being the convex program known as \emph{Basis Pursuit Denoising} (BPDN),
\begin{equation}
\label{eq:BPDN}
\widehat{\bs x} = \min_{\bs x \in \mathbb{R}^{N}}\|\bs x\|_{1}~~\mbox{s.t.}~~\|\bs y - \Phi \bs x\|_{2} < \epsilon,
\end{equation}
that guarantees  $\|\bs x - \widehat{\bs x}\|_{2} \leq C\epsilon$ for $\|\bs e\|_{2} < \epsilon$, $C$ a constant, and under certain conditions on $\Phi$~\cite{CandesDLP}.
A significant body of work has been devoted to the study of each of these components individually, e.g., by \emph{a}) characterizing conditions on $\Phi$ that provide robust mappings of sparse signals and designing physical sampling systems that satisfy such conditions~\cite{TroppG_Signal,CandesDLP,TroLasDua::2009::Beyond-Nyquist:,DuaDavTak::2008::Single-pixel-imaging, SlaviLDB_Compressive,BajwaHRWN_Toeplitz}; \emph{b}) proposing more refined classes of highly structured signals~\cite{Duarte_spectral,HegdeDC_Compressive,BaranCDH_Model}; and  \emph{c}) providing reconstruction guarantees and fast solvers for BPDN and other convex programs \cite{HaleYZ_Fixed,YinOGD_Bregman,FigueNW_Gradient,BergF_Probing} as well as greedy and first-order algorithms~\cite{cosamp,BluDav::2008::Iterative-hard,DonohMM_Message}.

CS promises to lessen our sampling burden.  The simple consequence of (\ref{eq:ypx}) is that, when the acquisition of each measurement is ``expensive,'' we benefit by sensing only $M$ values rather than $N$.  One example of such a situation is magnetic resonance imaging (MRI)~\cite{LustiDP_Rapid}.  We seek to minimize the amount of time to image a patient; however, each measurement is time-consuming, leading to a total acquisition time that is currently on the order of tens of minutes.  Another example is sampling, where the required Nyquist rate for wideband signals may be prohibitively highy~\cite{Hea05:Analog-to-Information,TroLasDua::2009::Beyond-Nyquist:}.  It is possible to design a physical sampling system $\bar{\Phi}$ such that $\bs y = \Phi\bs x = \bar{\Phi}(x(t))$ where $\bs x$ is a vector of Nyquist-rate samples of a bandlimited signal $x(t)$, $t \in\mathbb{R}$.  In this case, (\ref{eq:ypx}) translates to low, sub-Nyquist sampling rates, a potential boon for wideband acquisition.

In practice, three issues may arise during signal acquisition that are not modeled by (\ref{eq:ypx}).  First, the real-valued CS measurements will be mapped to discrete bits via a quantizer.  Second, there may be noise present on the input signal. Third, we often must limit the total number of measured bits $\mathfrak{B}$, i.e., we are constrained by a bit-budget when transmitting or storing the measurements.
Thus, a more precise model of CS acquisition is 
\begin{equation}
\label{eq:bettermodel}
\bs y_{Q} = \mathcal{Q}_{B}(\Phi(\bs x+ \bs n) + \bs e),
\end{equation}
where the \emph{signal noise} is denoted by $\bs n\in \mathbb{R}^{N}$, and $\mathcal{Q}_{B}: \mathbb{R} \rightarrow \mathfrak{A}$ is a $B$-bit scalar quantization function (applied element-wise in (\ref{eq:bettermodel})) that maps real-valued CS measurements to the discrete alphabet $\mathfrak{A}$ with $|\mathfrak{A}| = 2^{B}$. In this paper we will model $\bs n$ as a random vector with each element having variance $\sigma_{\bs n}^{2}$.  Since the primary source of measurement noise in a well-designed hardware system derives from quantization, we will assume $\|\bs e\|_{2}=0$.\footnote{The general trends presented in this paper remain unchanged when $\bs \|e\|_{2} > 0$.}  Since the quantizer is scalar, we can write the bit-budget constraint as 
\begin{equation}
\mathfrak{B} = MB.
\end{equation}
Although we will focus on scalar quantization in this paper, alternative quantization techniques such as sigma-delta~\cite{GunPowSaa::2010::Sobolev-Duals} or non-monotonic scalar quantization~\cite{Boufounos::2010::univer_rate_effic_scalar_quant} have also been proposed for CS systems,  as have many algorithms specialized to CS quantization problems~\cite{bib:GPSY_SD10,ZymBoyCan::2009::Compressed-sensing,SunGoy::2009::Quantization-for-compressed,SunGoy::20090::Optimal-quantization,vivekQuantFrame,LasBouDav::2009::Demcracy-in-action:,JacquHF_Dequantizing}.  The main themes presented here will be generally applicable to these techniques and algorithms as well.

The fixed bit-budget $\mathfrak{B}  = MB$ and the signal noise $\bs n$ impose a competing performance tradeoff as a function of $M$.  On the one hand, since $B = \mathfrak{B}/M$, we can increase the bit-depth as we decrease the number of measurements, thereby increasing the precision of each measurement.  On the other hand, signal noise is amplified due to \emph{noise folding} as we decrease the number of measurements, thereby decreasing the precision of each measurement~\cite{CanDav::2011::How-well-can-we-estimate}.\footnote{Roughly speaking, noise folding implies that during reconstruction we lose about $3$dB of signal-to-noise ratio (SNR) as the number of measurements is halved~\cite{DASP,DavLasTre::2011::The-pros-and-cons}.}  Thus, we find ourselves in somewhat of a conundrum: as we take fewer measurements we can allocate more bits per measurement (good), but noise folding increases the risk of wasting these bits on already imprecise measurements (bad).

We can gain more insight into this conundrum through a back-of-the-envelope calculation of the optimal total acquisition error, which comprises the expected mean-squared distortion due to a scalar quantizer for Gaussian measurements $O(\|x\|_{2}^{2}2^{-2B})$ and the expected reconstruction error due to measurement noise $O\left(\frac{N}{M}\sigma_{\bs n}^{2}\right)$.  Equating these noise levels to minimize the total mean square error (MSE) leads to
\begin{equation}
\label{eq:env}
B \approx \frac{1}{2}\log_{2}\left( \frac{\|\bs x\|_{2}^{2}}{\sigma_{\bs n}^{2}} \frac{M}{N}\right).\nonumber
\end{equation}
This expression can also be found using classical rate-distortion bounds in terms of the signal-to-noise ratio (SNR)~\cite{CoverT_Elements,SarvoBB_Measurements}.  Imposing the fixed bit-budget $B = \mathfrak{B}/M$ and rearranging terms, we find that the MSE is minimized when
\begin{equation}
\label{eq:env2}
\log_{2}\left(\frac{\|\bs x\|_{2}^{2}}{N\sigma_{\bs n}^{2}}\right) \approx \frac{2\mathfrak{B}}{M} - \log_{2}\left(M \right).
\end{equation}
The term on the left is the logarithm of the SNR of the input signal.  
For fixed $\mathfrak{B}$ and $N$, (\ref{eq:env2}) implies that there are two operational regimes that correspond roughly to ``high'' input SNR and ``low'' input SNR.  At high input SNR, the MSE is minimized by taking a small number of measurements $M$ with large bit-depth; we call this the  \emph{measurement compression} (MC) regime.  At low input SNR, the MSE is minimized by taking a large number of measurements $M$ with small bit-depth; we call this the \emph{quantization compression} (QC) regime.  The exact SNR at which the transition between the two regimes occurs is a function of the total bit-budget.  A primary contribution of this paper is to expose and explore the QC regime.

In this paper we argue for the distinction between the MC and QC regimes in two ways.  First, we formalize the back-of-the-envelope calculation in (\ref{eq:env}) by analyzing the reconstruction MSE that results from the combined effects of quantization and signal noise folding.  Specifically we provide an upper bound on this MSE for an optimal non-uniform scalar quantizer that roughly predicts the trends of the optimal bit-depth for different signal noise powers and bit-budgets.  Second,  we provide a suite of simulations for a specific setup frequently encountered in practice:  the acquisition of sparse signals from uniformly quantized measurements.  Surprisingly, at certain practical SNRs, our simulations suggest that a $1$-bit quantizer (using the reconstruction techniques developed in \cite{JacLasBou::2011::Robust-1-bit}) exhibits better performance than larger bit-depth quantizers.

Revisiting the example CS applications from above, a CS MRI device should aim to operate in the MC regime, since the total data acquisition time is proportional to $M$.  In this case, (\ref{eq:env}) recommends acquiring high SNR measurements and quantizing them finely.   In contrast, a low SNR wideband sampling system should aim to operate in the QC regime.  In this case, (\ref{eq:env}) recommends acquiring low SNR measurements and quantizing them coarsely.  Fortunately, by some divine Providence, sampling rate and bit-depth enjoy an inverse relationship in practical ADCs;  specifically, we obtain an exponential increase in sampling rate as the bit-depth is decreased~\cite{LeRonRee::2005::Analog-to-Digital-Converters}.   Taking this idea to its logical extreme, it has been shown that it is possible to drive the bit-depth down to $1$ bit per CS measurement and still guarantee stable signal recovery~\cite{BouBar::2008::1-Bit-compressive,LasWenYin::2010::Trust-but-verify:,Bou::2009::Greedy-sparse,JacLasBou::2011::Robust-1-bit}. In this case the quantizer is simply a comparator, enabling an extremely high sampling rate.

The remainder of this paper is organized as follows.  In Section~\ref{sec:background}, we provide the necessary CS background for our analysis and simulations.  In Section~\ref{sec:anal}, we develop a bound on the reconstruction error due to quantization and signal noise, expressed in terms of a fixed bit-budget.  In Section~\ref{sec:sims}, we present a series of numerical simulations that further support our argument.  We conclude in Section~\ref{sec:disc} with a discussion on the implications of this work.


\section{Background}
\label{sec:background}
\subsection{CS Toolkit}\label{sec:toolkit}
Before examining the effect of noise and quantization on CS reconstruction performance, we first review a few key results and definitions that enable our analysis.

CS reconstruction can be interpreted as consisting of two steps:  first finding the non-zero coefficient locations (the support) and then estimating the coefficient values.  If we can correctly identify the true signal support, then the optimal linear estimate for coefficient values can be computed via least squares:
\begin{equation}
\label{eq:oracle}
\widehat{\bs x}|_{\Omega} =  \Phi_{\Omega}^{\dagger}\bs y, \quad
\widehat{\bs x}|_{\Omega^C} =  \bs 0,
\end{equation}
where $ \Phi_{\Omega}$ denotes the submatrix of $\Phi$ formed by selecting the columns of $\Phi$ according to the index set $\Omega$, $\widehat{\bs x}|_{\Omega}$ is the corresponding subvector  of $\widehat{x}$, $\Omega^{C}$ is the complement set to $\Omega$, and $\dagger$ denotes the Moore-Penrose pseudo-inverse.  Indeed, if an oracle were to provide the true support $\Omega$, then no linear CS reconstruction algorithm can perform better than (\ref{eq:oracle}). Thus, reconstruction with known signal support is sometimes called \emph{oracle-assisted} reconstruction~\cite{CandeT_Dantzig,DavLasTre::2011::The-pros-and-cons}.  Our analysis will be primarily in terms of the performance of this best-case reconstruction algorithm.
Furthermore, from \cite{JacLasBou::2011::Robust-1-bit,Boufounos::2010::univer_rate_effic_scalar_quant}, when there is no noise on the measurements, the
reconstruction (\ref{eq:oracle}) is also consistent, meaning that
$$\cl
Q_B(\Phi \bs{\hat x}) = \cl Q_B(\Phi_\Omega \bs{\hat
x}|_{\Omega})=\cl Q_B(\Phi_\Omega\Phi^\dagger_\Omega \bs y) = \cl
Q(\bs y) = \bs y.
$$
There is no better nonlinear estimator for the quantized measurements than a consistent estimator.

Robust reconstruction guarantees will only hold for measurement systems $\Phi$ that are ``well-conditioned.''   For instance, the so-called \emph{restricted isometry property} (RIP) of a matrix $\Phi$ has been shown to be a sufficient condition for the robust recovery of sparse signals via several algorithms~\cite{CandesDLP,cosamp}.  The RIP of order $K$ with constant $\delta$ is defined as 
\begin{equation}
(1-\delta)\|\bs x\|_{2}^{2} \leq \|\Phi \bs x\|_{2}^{2} \leq (1+\delta)\|\bs x\|_{2}^{2},
\end{equation}
for all $\bs x \in \Sigma_{K}$.   Roughly speaking the RIP ensures that the norm of the measurements is close to the norm of the signal for all $K$-sparse signals.  An alternative way of thinking of this is that the singular values of any submatrix formed by $K$ or fewer columns of $\Phi$ are bounded close to $1$; hence any $K$-column submatrix of $\Phi$ is close to an isometry.

The RIP ensures stable oracle-assisted recovery when white noise is added to the measurements.  Specifically, suppose that $\bs z = \Phi \bs x - \bs y$, where $\bs z$ is a zero-mean random vector with uncorrelated (white) entries, each having variance $\sigma_{\bs z}^{2}$.  Furthermore suppose that $\Phi$ has the RIP of order $K$, and that $\bs x$ is $K$-sparse.  Then Theorem 4.1 of \cite{DavLasTre::2011::The-pros-and-cons} demonstrates that oracle-assisted reconstruction will have expected error 
\begin{equation}
\label{eq:oraclerecon}
\frac{K\sigma_{\bs z}^{2}}{1+\delta} \leq \mathbb{E}(\|\bs x - \widehat{\bs x}\|_{2}^{2}) \leq \frac{K\sigma_{\bs z}^{2}}{1-\delta}.
\end{equation}
A key component of our analysis below will be understanding the variance of the noise term $\bs z$ that arises from the quantized noisy measurements $\bs y_{Q}$.  The expression (\ref{eq:oraclerecon}) then gives the intuition that the expected reconstruction error behaves on the order of the variance of the error per measurement $\sigma_{\bs z}^{2}$.

We will also make use of a result that relates the variance
$\sigma_{\bs n}^{2}$ of the signal noise to the variance of the
measured noise $\sigma_{\Phi \bs n}^{2}$.  If $\bs n$ is white with mean zero and
variance $\sigma_{\bs n}^{2}$, and $\Phi$ has orthonormal rows, i.e., $\Phi\Phi^{T} = \frac{N}{M}\bs I_{M}$,\footnote{The so-called \emph{tight frame} condition $\Phi\Phi^{T} = \frac{N}{M}\bs I_{M}$ is not overly restrictive, since for any RIP matrix $\Gamma$, a matrix that has both the same row-space as $\Gamma$ and the tight frame condition can be derived from $\Gamma$~\cite{DavLasTre::2011::The-pros-and-cons}.}
then it is straightforward to show
that the measured noise is also white and zero mean and has variance
\begin{equation}
\label{eq:noisefold}
\sigma_{\Phi \bs n}^{2} = \frac{N}{M}\sigma_{\bs n}^{2}.
\end{equation}
Note that the measured noise is only uncorrelated (i.e., white) when $M\leq N$;  indeed, the condition $\Phi\Phi^{T} = \frac{N}{M}\bs I_{M}$ can only hold when $M\leq N$.

In \cite{DavLasTre::2011::The-pros-and-cons}, the authors
combine the results of (\ref{eq:oraclerecon}) and (\ref{eq:noisefold}) to obtain a bound on the
oracle-assisted reconstruction error due to noise folding.  
We will take a similar approach, however we will additionally include the effects of quantization.  Furthermore, because our quantization error is not necessarily uncorrelated, we first generalize (\ref{eq:oraclerecon}) to obtain an upper bound on the oracle reconstruction error with uncorrelated measurement noise.

\subsection{$1$-bit CS}
The results of the conventional CS framework above will enable us to analyze scalar quantized measurements when the bit-depth is greater than $1$.
However, CS measurements can be coarsely quantized to just $1$ bit, representing their signs. 
These facts preclude $1$-bit CS from being analyzed within the
conventional linear CS framework.
 Even though meaningful theoretical comparisons are difficult to make between $1$-bit and conventional CS, it is beneficial to compare their empirical performances, since both types of CS can be useful in practice.
Thus, we very briefly review the key
results of the $1$-bit CS
framework~\cite{BouBar::2008::1-Bit-compressive,JacLasBou::2011::Robust-1-bit}.
Formally, $1$-bit measurements can be written as
\begin{equation}
\label{eq:defh}
\bs y_{s} = A(\bs x) := \sign (\Phi \bs x).
\end{equation}
To reconstruct, we search for a sparse, unit-norm signal $\bs{\hat x}$
that is \emph{consistent} with the measurements, meaning that $A(\bs{\hat x}) =
A(\bs x)$.  We restrict our attention to unit-norm signals, since the
scale of the signal is lost during the $1$-bit quantization
process.
This problem is generally non-convex, and thus it is difficult to
design an algorithm that will be guaranteed to find the desired
solution. Nonetheless several algorithms have been proposed to
approximately solve this
problem~\cite{BouBar::2008::1-Bit-compressive,Bou::2009::Greedy-sparse,LasWenYin::2010::Trust-but-verify:,JacLasBou::2011::Robust-1-bit}; convex programs have also been formulated~\cite{PlaVer::2011::One-bit-compressed}.

In much the same way that the RIP of $\Phi$ guarantees stable reconstruction from
$\ell_{1}$-minimization programs~\cite{CandeT_Decoding}, the so-called
\emph{binary $\epsilon$-stable embedding} (B$\epsilon$SE) provides a
similar robustness for the mapping $A$ with consistent
algorithms~\cite{JacLasBou::2011::Robust-1-bit}.  The property
explains that the normalized Hamming distance between any two sets of
measurements is within $\epsilon$ of the normalized angular distance
between the original signals, for all unit-norm $K$-sparse signals.
It can be shown that, if the elements of $\Phi$ are drawn from a
Gaussian distribution, then $\Phi$ satisfies the B$\epsilon$SE with
high probability, and thus CS systems that enable $1$-bit
quantized measurements exist.

Our simulations will make use of two $1$-bit CS algorithms originally introduced in~\cite{JacLasBou::2011::Robust-1-bit}.  Specifically, we will employ the BIHT and BIHT-$\ell_{2}$ algorithms.  The former can be thought of as minimizing a one-sided $\ell_{1}$-norm and imposing a sparse unit-norm signal model, while the latter can be thought of as minimizing a one-sided $\ell_{2}$-norm instead.  By one-sided norm, we mean that the positive elements of a vector are set to zero before the norm is computed. The BIHT algorithm has been shown to perform better in low noise scenarios, while the BIHT-$\ell_{2}$ algorithm has been shown to perform better in high noise scenarios~\cite{JacLasBou::2011::Robust-1-bit}.



\section{Analysis of Quantized CS Systems with Signal Noise}
\label{sec:anal}

In this section we derive a new upper bound on the oracle-assisted reconstruction error due to both noise and quantization, making the back of the envelope calculation (\ref{eq:env}) more rigorous.
This bound enables us to argue that, for a fixed bit-budget $\mathfrak{B} = MB$, it may be better to quantize to fewer bits per measurement $B$ than take fewer measurements $M$. The following theorem is proved in Appendix~\ref{apx:proof}.

\begin{theorem}
\label{thm:reconerrorbound}
Suppose that $\bs y_{Q} =  \mathcal{Q}_{B}(\Phi(\bs x + \bs n))$.  Let the signal $\bs x \in \mathbb{R}^{N}$ be sparse with support $\Omega \in \{1, \ldots, N\}$ and $|\Omega|=K$, where the elements $\Omega$ are chosen uniformly at random and the amplitudes of the non-zero coefficients are drawn according to $x_{j}\in \Omega \sim \mathcal{N}(0,\sigma_{\bs x}^{2})$.  Let the signal noise $\bs n \in \mathbb{R}^{M}$ be a random, white, zero-mean vector with variance $\sigma_{\bs n}^{2}$.  
Furthermore, let the $M\times N$ matrix $\Phi$ satisfy the RIP of order $K$ with constant $\delta$, $\Phi\Phi^{T} = \frac{N}{M}\bs I_{M}$, and $M<N$.  Choose $\mathcal{Q}_{B}$ to be the optimal scalar quantizer with $B>1$ that minimizes the MSE for the distribution of the measurements  $\Phi(\bs x + \bs n)$.
Then for a fixed bit-budget of $\mathfrak{B} = MB$,  the MSE of the oracle-assisted reconstruction estimate $\widehat{\bs x}$ satisfies
\begin{equation}
\label{eq:errbound}
\mathbb{E}\left(\| \bs x - \widehat{\bs x}\|_{2}^{2}\right) \leq \frac{2K}{\mathfrak{B}(1-\delta)}\left( K\sigma_{\bs x}^{2}B2^{-2B}  + N\sigma_{\bs n}^{2} B\left(1+2^{-2B}\right) \right) + \frac{K}{(1-\delta)} \left(\frac{\mathfrak{B}}{B}-1\right) \mathfrak{S},
\end{equation}
where $\mathfrak{S} = \max_{i\neq j}|\mathbb{E}(\mathcal{Q}_{B}(\Phi \bs x + \Phi \bs n)_{i}\mathcal{Q}_{B}(\Phi \bs x + \Phi \bs n)_{j})|$ is the correlation between the quantized measurements.
\end{theorem}

Each component of the bound  (\ref{eq:errbound})  is fairly intuitive.  The term $K\sigma_{\bs x}^{2}B2^{-2B}$ reflects the error due to quantizing the measurements.    The term $N\sigma_{\bs n}^{2} B\left(2^{-2B} + 1\right)$  reflects both the error due to measured signal noise as well as the quantization of that noise. The reconstruction error is effectively proportional to these two terms.   The final term $\left(\frac{\mathfrak{B}}{B}-1\right)\mathfrak{S}$ reflects an additional error due to the correlation between the quantized measurements.  In many CS scenarios we expect this term to be close to zero, and furthermore for large $B$ it has been shown that this term can be accurately approximated as zero~\cite{GraNeu::1998::Quantization}.  Thus, choosing the optimal $B$ primarily comes down to balancing the terms inside the parentheses.

The bound in (\ref{eq:errbound}) applies to strictly sparse signals immersed in signal noise.  However, it may also be of interest to consider so-called \emph{compressible signals}, i.e., signals that are not strictly sparse but that can be reasonably approximated by retaining their $K$ largest magnitude coefficients. For such signals, the ``tail'' part of the signal that we do no expect to recover, i.e., the subset of the smallest $N-K$ entries, is also subject to noise folding. Theorem~\ref{thm:reconerrorbound} can be extended to handle compressible signals by inflating the second term to account for the additional correlation between the quantized measurements. The general performance trends will be similar to sparse signals in noise; i.e., signals that are ``less compressible'' will induce the same regime as signals with low input SNR.

The bound in (\ref{eq:errbound}) is pessimistic, since we do not take into account the benefits accrued by increasing the number of measurements, for instance by improving the RIP constants of $\Phi$.  
Furthermore, when the quantization error is large enough to dominate the measurement noise, the measurement noise terms may not play an active role in the true behavior of the system.
Again, this is not reflected by the bound.  Finally, the bound does not apply to $1$-bit quantization or the case where $M > N$.

\begin{figure*}[!t] 
   \centering
   \begin{tabular}{cc}
   \includegraphics[width=.84\imgwidth]{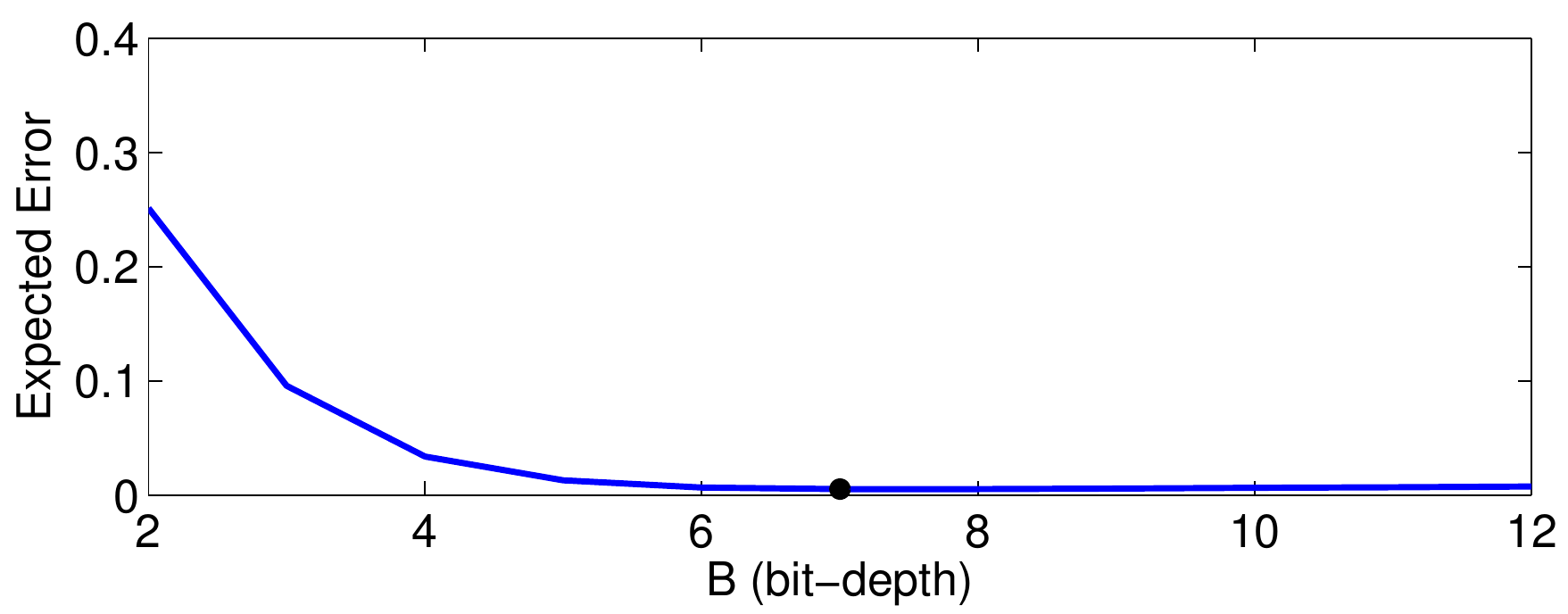}&
       \includegraphics[width=.84\imgwidth]{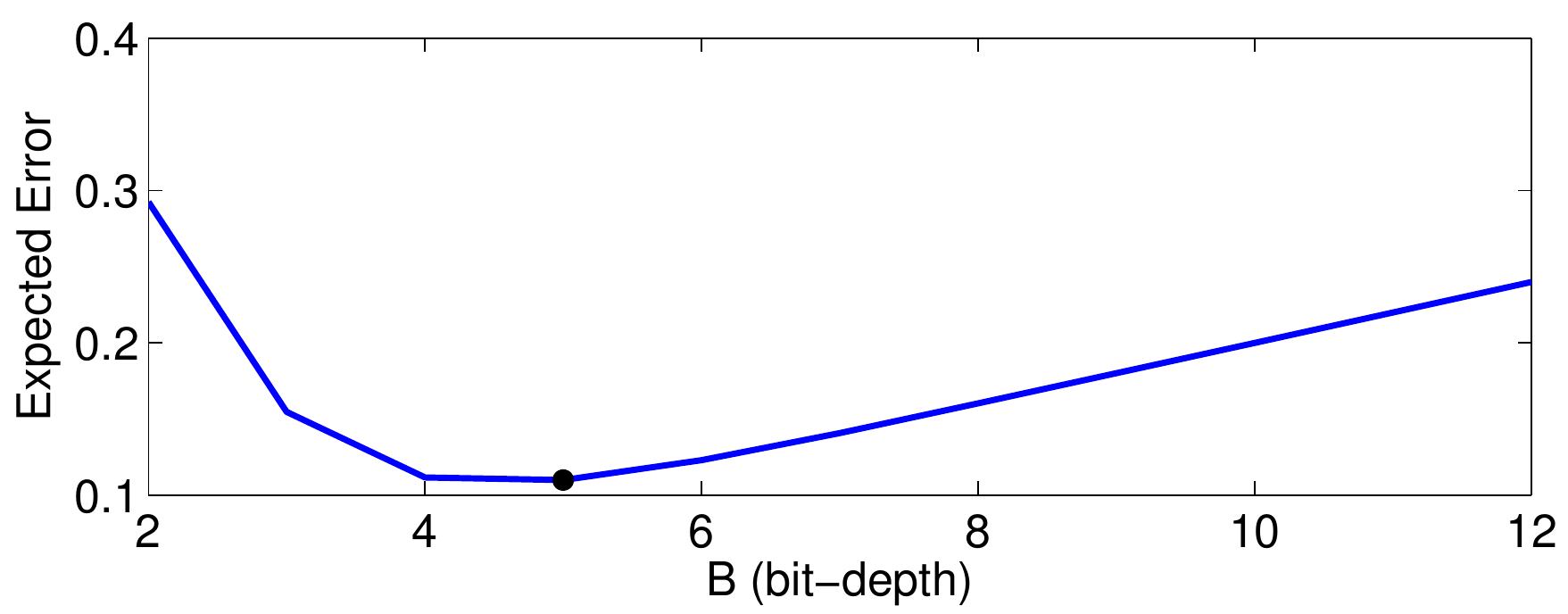}\\
       (a) \small{$\mathrm{ISNR} = 35$dB, optimal bit-depth $= 7$}&(b) \small{$\mathrm{ISNR} = 20$dB, optimal bit-depth $= 5$}\\
          \includegraphics[width=.9\imgwidth]{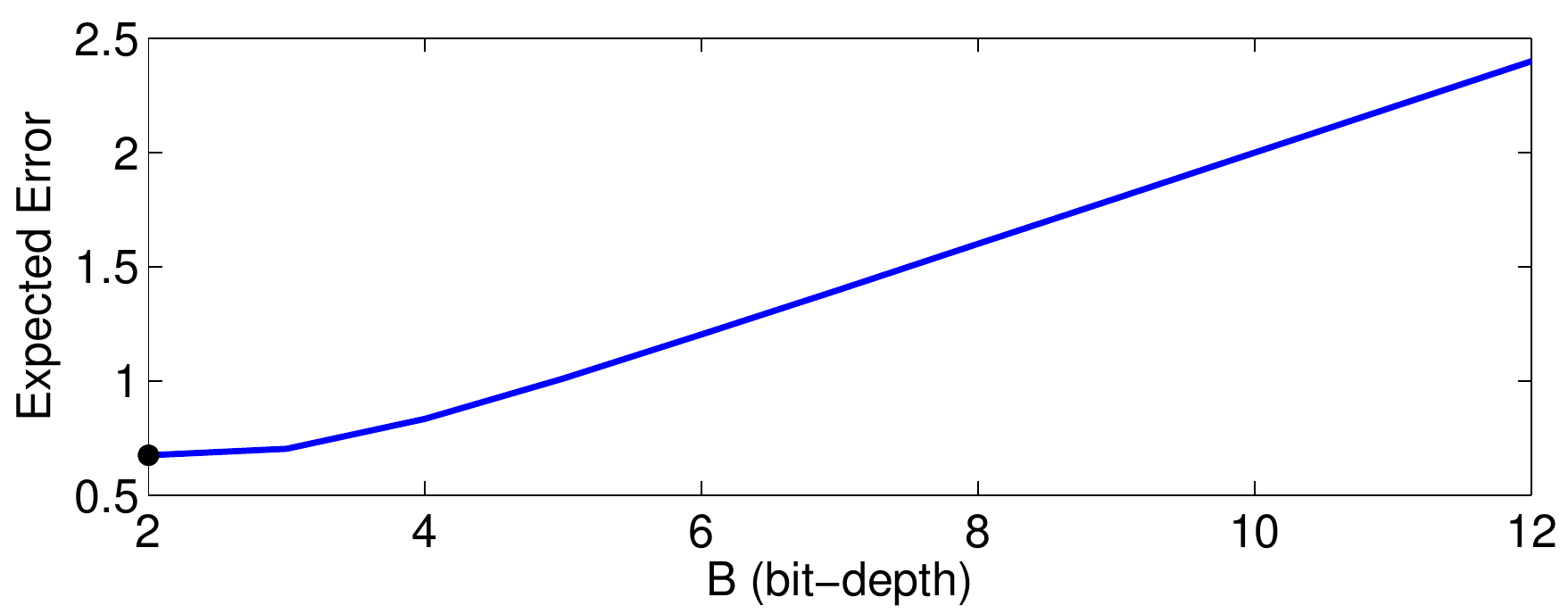}&
   \includegraphics[width=.87\imgwidth]{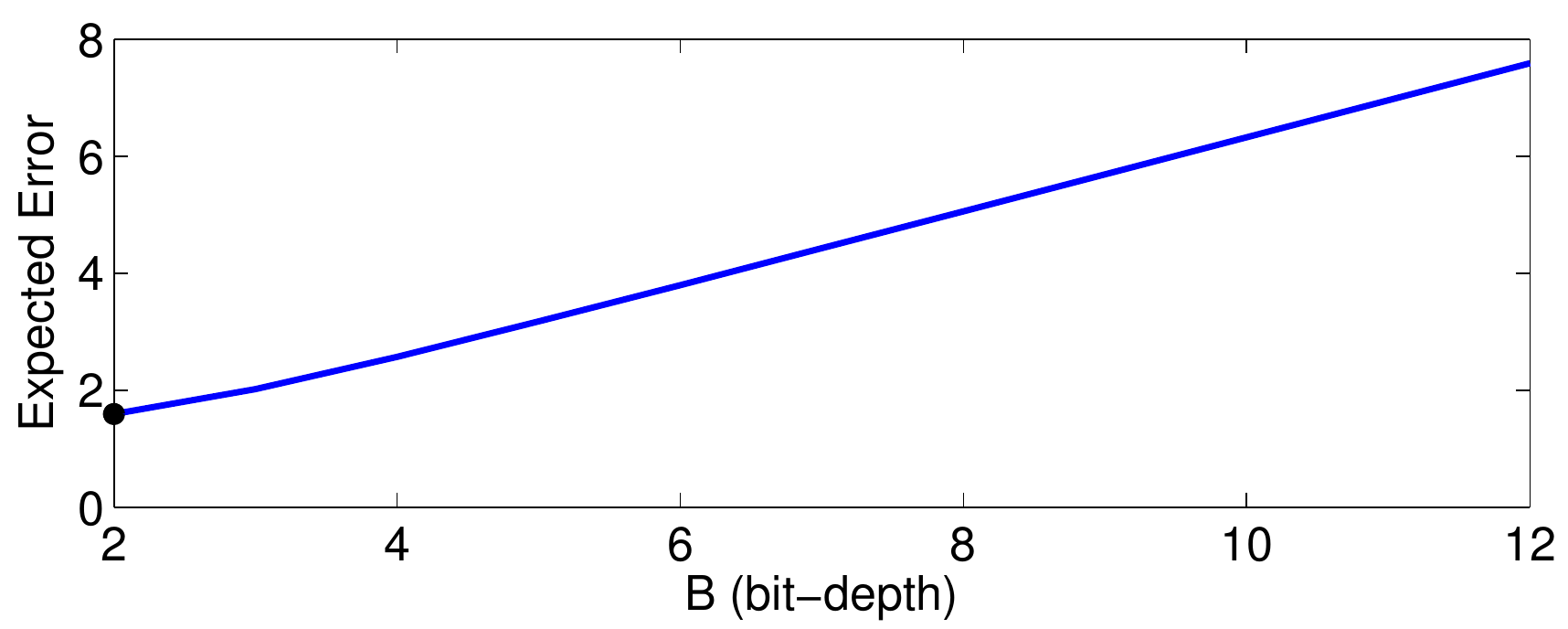}\\
	(c) \small{$\mathrm{ISNR} = 10$dB, optimal bit-depth $= 2$}&(d)\small{ $\mathrm{ISNR} = 5$dB, optimal bit-depth $= 2$}\\
   \end{tabular}
   \caption{Upper bound on the oracle-assisted reconstruction error as a function of bit-depth $B$ and $\mathrm{ISNR}$.  The term inside the parenthesis in the bound (\ref{eq:errbound}) was computed.  
Black dots denote the minimum point on each curve. 
   }
   \label{fig:bound}
\end{figure*}
To use the bound (\ref{eq:errbound}) to support our argument that there are both MC and QC regimes in CS,  we examine the behavior of the oracle-assisted reconstruction error as a function of the bit-depth $B$ (or equivalently the number of measurements $M$ since $\mathfrak{B} = MB$).  
Since the solution for the optimal $B$ cannot be computed in closed form without resorting to tabulated functions, we evaluate the bound over some interesting parameters.  The evaluation of the bound is depicted in Figure~\ref{fig:bound}, where plots (a)--(d) correspond to input signal-to-noise ratios (ISNRs) of $35$dB, $20$dB, $10$dB, and $5$dB, respectively.  We define the \emph{input SNR} (ISNR) in dB as 
\begin{equation}
\label{eq:isnr}
\mathrm{ISNR} := 10\log_{10}\left(\frac{\mathbb{E}(\|\bs x\|_{2}^{2})}{\mathbb{E}(\|\bs n\|_{2}^{2})}\right).
\end{equation}
where $\mathbb{E}(\|\bs x\|_{2}^{2}) = K\sigma_{\bs x}^{2}$ and $\mathbb{E}(\|\bs n\|_{2}^{2}) = N\sigma_{\bs n}^{2}$.  

Since we are primarily concerned with the performance trend of (\ref{eq:errbound}) as a function of $B$ and the ISNR, we make a few simplifications when plotting the bound.  First, we only evaluate the term inside the parenthesis; this term is proportional to the error on the measurements and does not depend on the RIP constant, the sparsity $K$, or the correlation between the quantization errors.  Second, by only evaluating the term inside the parenthesis in (\ref{eq:errbound}), we do not take into account the effect of $M$ on the RIP constants ($\delta$ decreases as $M$ increases).    
 The minimum error point in each curve is denoted by a solid black dot.  

The message from Figure~\ref{fig:bound} is clear.  The tradeoff between the number of measurements $M$ and bit-depth $B$ empirically follows a convex curve, i.e., the error not only increases when $B$ is too small, but the error also increases when $B$ is too large.  In other words, more bits per measurement is not necessarily optimal.   Furthermore,  as expected, the minimum reconstruction error occurs for smaller $B$ as the ISNR decreases.  For the high ISNR of $35$dB, the bound is minimized at a bit-depth of approximately $7$ bits per measurement.  The is an example of the MC regime, where larger bit-depths and thus lower $M$ yield the best performance.
For the low ISNR of $10$dB, the bound is minimized at a bit-depth of approximately $2$ bits per measurement.  
This is an example of the QC regime, where larger bit-depths and thus higher $M$ yield the best performance.


\section{Experiments}
\label{sec:sims}
In the previous section we have argued that the QC regime exists by deriving an upper bound on the oracle-assisted reconstruction error.  In this section we perform a suite of simulations to empirically study for which input noise levels and bit-budgets this regime will occur in practical systems.
Specifically our simulations \emph{i}) validate the theoretical result in Theorem~1, \emph{ii}) demonstrate the performance achieved in practice when combining quantization and signal noise, and finally \emph{iii}) prove the existence of the QC regime.  A surprising additional result emerges from the simulations: when nontrivial signal noise is present, $1$-bit CS systems perform competitively with, if not better than conventional CS with uniform multibit quantization.

\subsection{Setup}
Our simulations were performed using canonically (identity) sparse signals $\bs x$.\footnote{The results of simulations did not change when the signals were DCT-sparse.}  The signals were measured with i.i.d.\ Gaussian matrices, i.e., $\bs y = \Phi(\bs x + \bs n)$ where the matrix $\Phi$ has elements $\phi_{i,j} \overset{\mathrm{i.i.d.}}{\sim} \mathcal{N}(0,1/M)$.   The measurements were quantized uniformly with quantization interval $\Delta = T2^{-B+1}$, where $T$ is the dynamic range of the quantizer.  In all simulations, we chose $T = \|\Phi \bs x\|_{\infty}$ to maximize the range of the quantizer and ensure that for any noiseless measurement $|(\Phi\bs x)_{i} - \mathcal{Q}_{B}((\Phi \bs x)_{i})| \leq \Delta/2$.

In each trial we drew a new $M\times N$ sensing matrix $\Phi$ and a new signal $\bs x$.  The non-zero coefficients of $\bs x$ were chosen according to a Gaussian distribution, and their positions were chosen at random.  We additionally added Gaussian noise to $\bs x$ to obtain the desired $\mathrm{ISNR}$.  For $B>1$, reconstruction of the estimate $\widehat{\bs x}$ was performed using the oracle-assisted reconstruction algorithm (\ref{eq:oraclerecon}) for Section~\ref{sec:oracle} and BPDN (\ref{eq:BPDN}) with an oracle value of $\epsilon = \|\bs y - \mathcal{Q}_{B}(\bs y)\|_{2}$ for the remaining subsections.  For $B=1$, reconstruction was performed using both the \emph{binary iterative hard thresholding} (BIHT-$\ell_{1}$) and BIHT-$\ell_{2}$ algorithms; the former generally performs better in lower noise scenarios and the latter performs better in higher noise scenarios~\cite{JacLasBou::2011::Robust-1-bit}. We report the \emph{reconstruction SNR} (RSNR)
\begin{equation}
\mathrm{RSNR} := 10\log_{10}\left( \frac{\|\bs x\|_{2}^{2}}{\|\bs x - \widehat{\bs x} \|_{2}^{2}} \right)
\end{equation}
in dB unless otherwise noted.
  Recall that the number of measurements and bit-depth are constrained by $\mathfrak{B} = MB$.  We average our results over 100 trials for each parameter tuple \sloppy{$(N,K, \mathfrak{B}, B,\mathrm{ISNR})$}.

\subsection{Oracle-assisted reconstruction}\label{sec:oracle}
\begin{figure*}[!t] 
   \centering
   \begin{tabular}{cc}
   \includegraphics[width=.84\imgwidth]{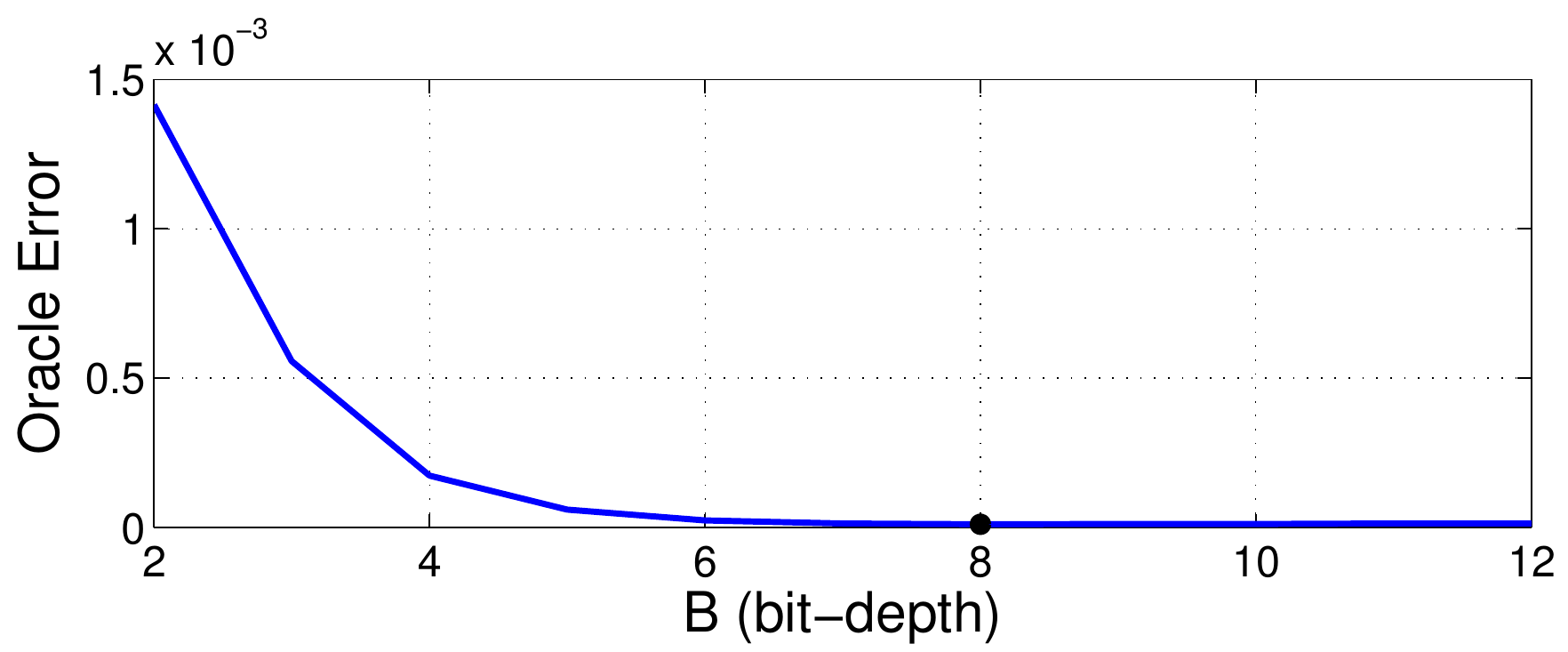}&
       \includegraphics[width=.84\imgwidth]{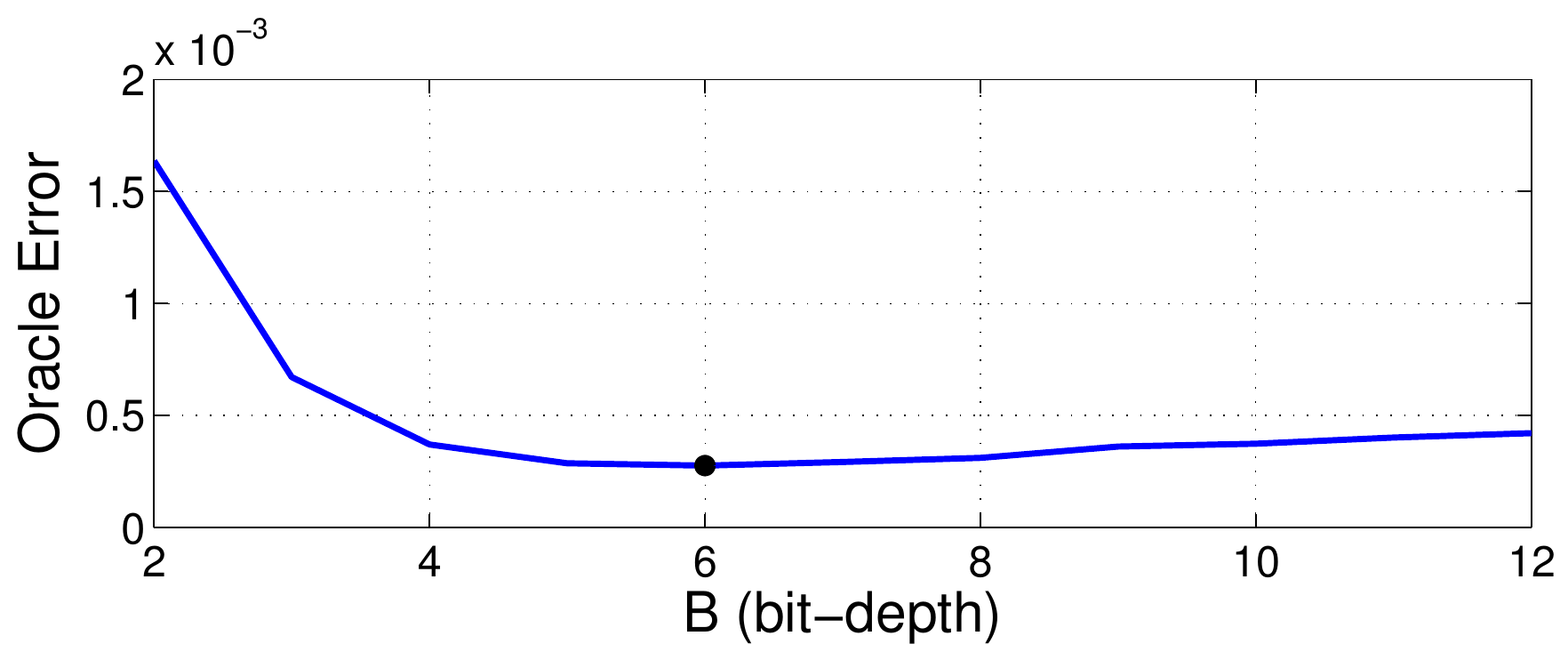}\\
       (a) \small{$\mathrm{ISNR} = 35$dB, optimal bit-depth $= 8$}&(b) \small{$\mathrm{ISNR} = 20$dB, optimal bit-depth $= 6$}\\
          \includegraphics[width=.84\imgwidth]{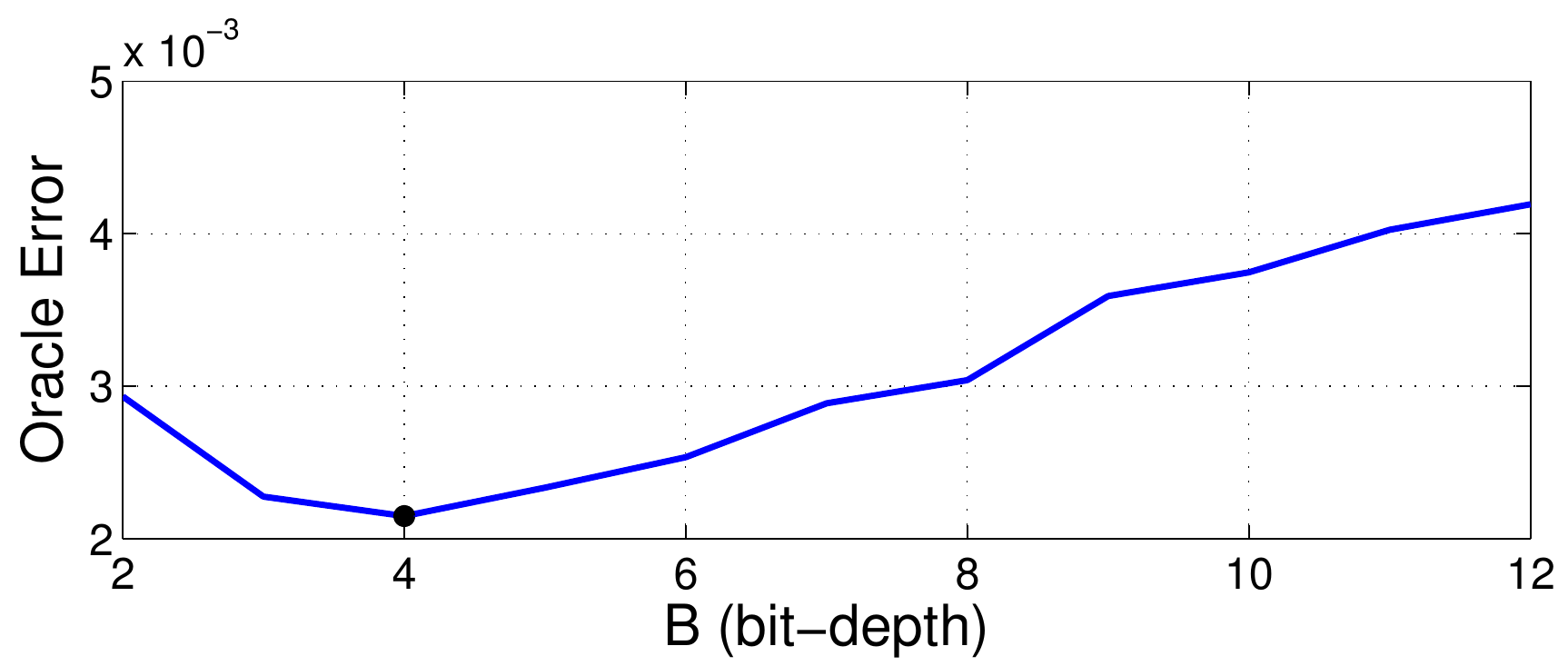}&
   \includegraphics[width=.87\imgwidth]{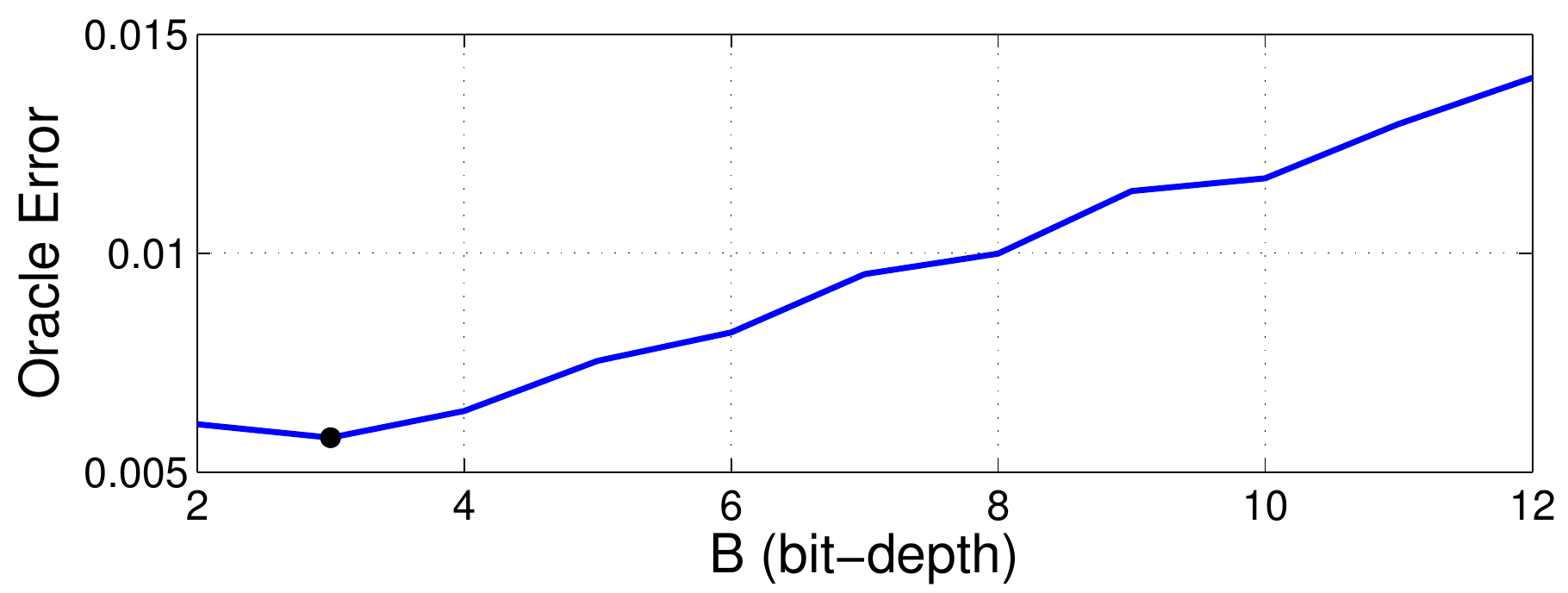}\\
	(c) \small{$\mathrm{ISNR} = 10$dB, optimal bit-depth $= 4$}&(d)\small{ $\mathrm{ISNR} = 5$dB, optimal bit-depth $= 3$}\\
   \end{tabular}
   \caption{Oracle-assisted reconstruction error (compare to the analytical upper bound plotted in Figure~\ref{fig:bound}) for $N=1000$, $K=10$, and $\mathfrak{B} = 3N$.  As predicted by (\ref{eq:errbound}), the minimum reconstruction error (denoted by black dots) is achieved by smaller bit-depths as the ISNR decreases. }
   \label{fig:errpermeas}
\end{figure*}

We begin by validating the message from Theorem~$1$, i.e., we examine the solution to the oracle-assisted reconstruction algorithm to see how the empirical performance relates to the bound (\ref{eq:errbound}).  Our goal is to compare the performance of our simulations to the theory-based plots in Figure~\ref{fig:bound}.  The experiments were performed as described previously with the oracle-assisted reconstruction algorithm. We plot the reconstruction error $\|\bs x - \widehat{\bs x}\|_{2}^{2}$ for bit-depths between $2$ and $12$ for a fixed bit-budget $\mathfrak{B}  = 3N$. We compared bit-depths of $2$ and higher, since (\ref{eq:errbound}) does not hold for lower bit-depths. Furthermore, unlike the statement of Theorem~\ref{thm:reconerrorbound}, recall that we used a uniform quantizer and not an optimal quantizer for the Gaussian measurements. Figures~\ref{fig:errpermeas}(a)--(d) depict the results for $\mathrm{ISNR} = 35$dB, $20$dB, $10$dB, and, $5$dB, respectively.   

The plots generally follow the same trends as in Figure~\ref{fig:bound}; however the minimum error occurs for a slightly higher bit-depth in each case.  The plots demonstrate that, as claimed in Section~\ref{sec:anal}, the best performance is obtained for smaller bit-depths as the ISNR decreases.

\subsection{Reconstruction performance as a function of $\mathfrak{B}$}
\begin{figure*}[!t] 
   \centering
   \begin{tabular}{cc}
   \includegraphics[width=.9\imgwidth]{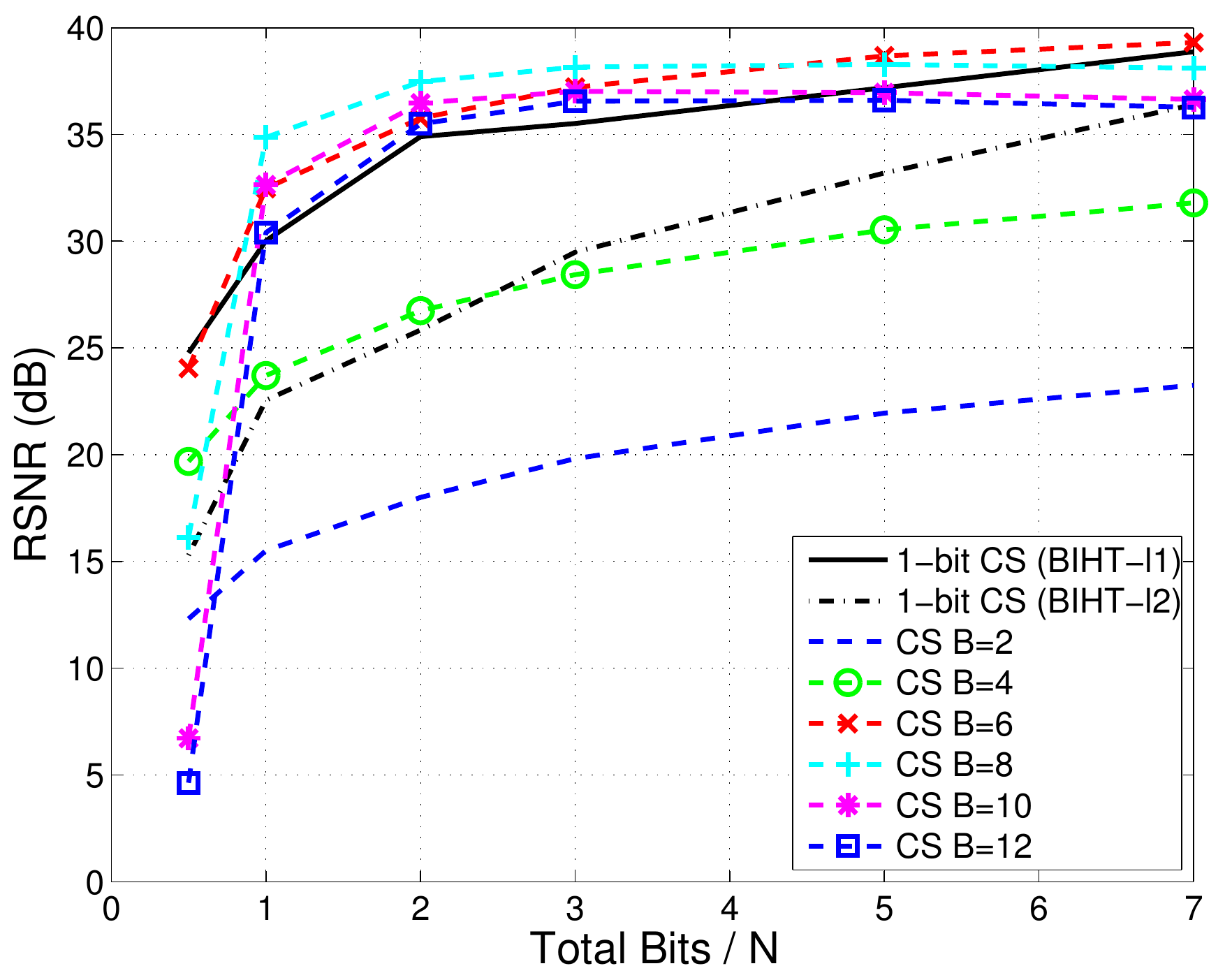}&
       \includegraphics[width=.9\imgwidth]{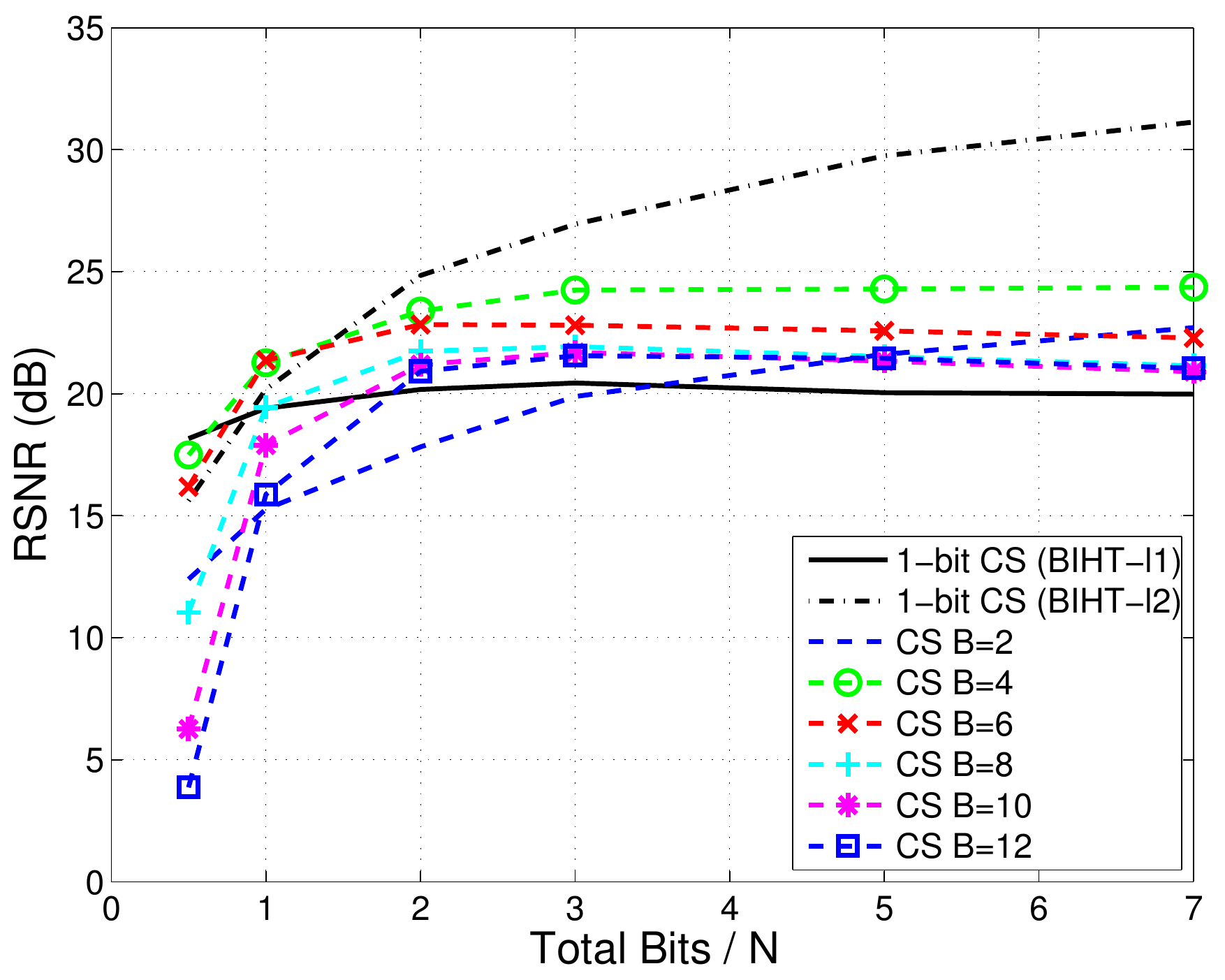}\\
       (a) $\mathrm{ISNR} = 35$dB &(b) $\mathrm{ISNR} = 20$dB\\
          \includegraphics[width=.9\imgwidth]{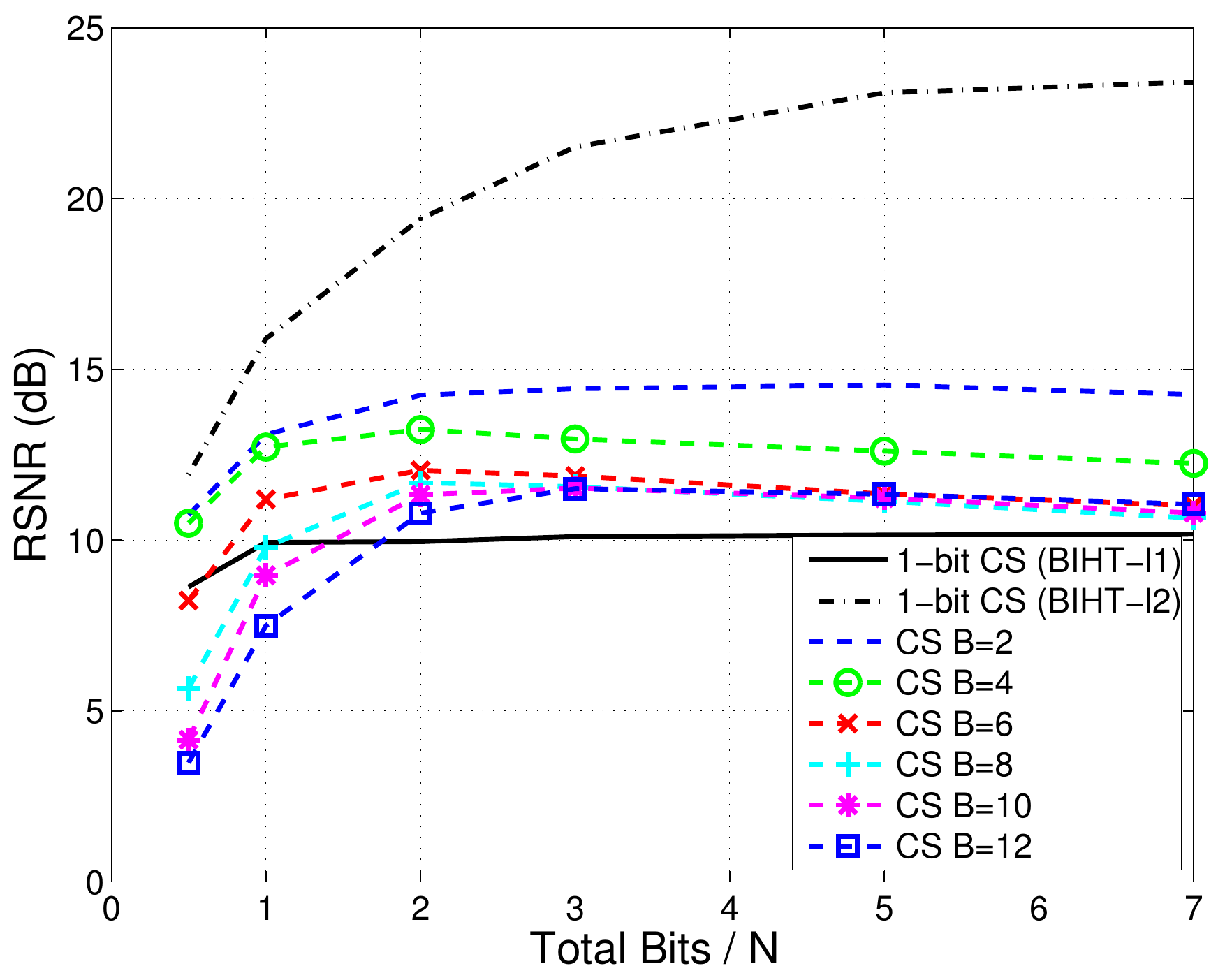}&
   \includegraphics[width=.9\imgwidth]{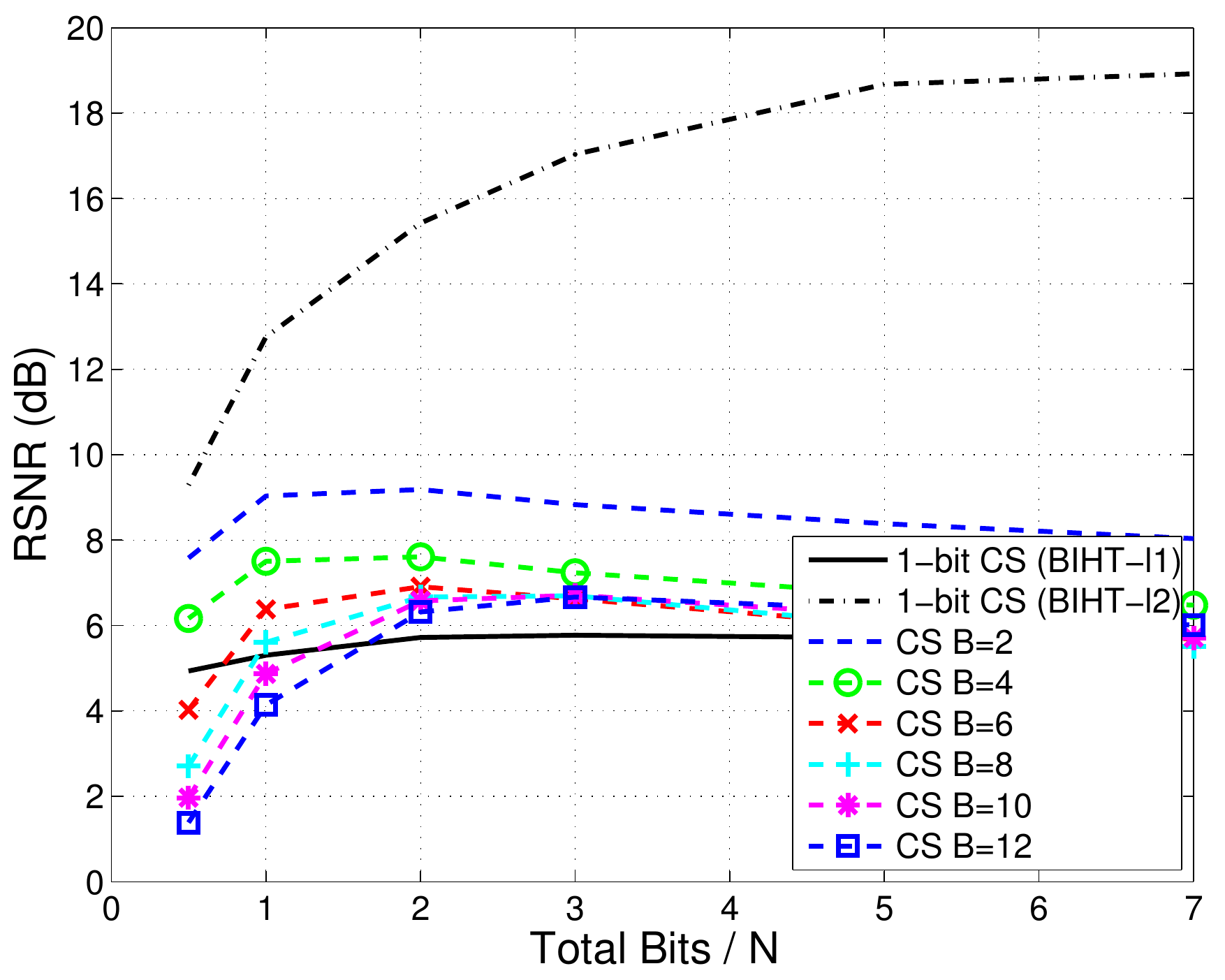}\\
	(c) $\mathrm{ISNR} = 10$dB &(d) $\mathrm{ISNR} = 5$dB\\
   \end{tabular}
   \caption{Reconstruction performance as a function of total bits, for different ISNRs.  Plots depict RSNR for different bit-depths $B$ for different $\mathrm{ISNR}$ with parameters $N=1000$ and $K=10$, and reconstruction via BPDN.   The figure demonstrates that as the ISNR is decreased, smaller bit-depths achieve better performance.  Additionally, $1$-bit CS techniques perform competitively with or better than BPDN for all ISNRs tested.}
   \label{fig:RSNRvB}
\end{figure*}
We next explore the performance achieved using practical algorithms instead of oracle-assisted reconstruction.  The experiments were performed as explained previously, for $N=1000$ and $K=10$, bit-depths $B = 1,2,4,6,8,10,12$, and for bit-budgets $\mathfrak{B} \in [N/2, 7N]$, with the BPDN and BIHT algorithms. Figures~\ref{fig:RSNRvB}(a)--(d) depict the experiment for the input $\mathrm{ISNR} = 35,20,10,5$dB, respectively.  

In the high ISNR regime of $35$dB, bit-depths of $B=1,6,8,10,$ and $12$ obtain similar RSNRs of around $35$dB, while smaller bit-depths result in poorer performance.  This is to be expected; since when the signal noise is fairly small, we will generally do better by using more bits per measurement.  

The performance of BIHT in this case is consistent with previous results showing that the $1$-bit techniques can outperform even $4$-bit uniformly quantized CS measurements with BPDN recovery.  This trend starts to reverse for lower signal ISNRs.  Indeed for ISNRs of $10$dB and $5$dB, we see that $2$ and $4$ bit-depth quantization outperforms larger bit-depths for all budgets.  Strikingly, the best performance for input SNRs of $20$dB, $10$dB, and $5$dB is achieved by acquiring just $1$ bit per measurement and reconstructing with the BIHT-$\ell_{2}$ algorithm.  

In addition to the simulations presented here, we also performed the similar simulations with $N=1000$ and $K=60$.  We found that all of the curves in Figure~\ref{fig:RSNRvB} dropped in SNR by roughly the same constant (that depends on $K$).  The relationship between the $1$-bit curves and the others was about the same for $\mathfrak{B} = 2N$ and lower.  For $\mathfrak{B} > 2N$, the $1$-bit reconstructions still outperformed the others; however the performance disparity was not as great as for $K=10$.

These simulations demonstrate two points.  First, they verify that the intuition provided by the upper bound (\ref{eq:errbound}) is indeed correct:  \emph{for lower ISNRs it is beneficial to choose smaller bit-depths $B$ and more measurements $M$.}  This validates the distinction between the QC and MC regimes.  Second, the $1$-bit CS setup performs significantly better than the multi-bit setup for low ISNRs and is competitive with the multi-bit setup for moderate ISNRs.  There are several reasons for this.  When the quantization error dominates the measurement noise, the reconstruction error is primarily due to the quantization error only. This case arises when $B$ is small; i.e., we can likely satisfy $\mathcal{Q}_{B}(\bs x + \bs n) = \mathcal{Q}_{B}(\bs x)$ for increasing values of $|n_{i}|$ as $B$ decreases.  
Furthermore, in this case consistent reconstruction of the $1$-bit algorithms may have an advantage.  Consistency could be presumably added to multibit reconstruction to improve performance but this is a topic left for future research.

\subsection{Reconstruction performance as a function of ISNR}

\begin{figure*}[!t] 
   \centering
   \begin{tabular}{ccc}
   \includegraphics[width=.55\imgwidth]{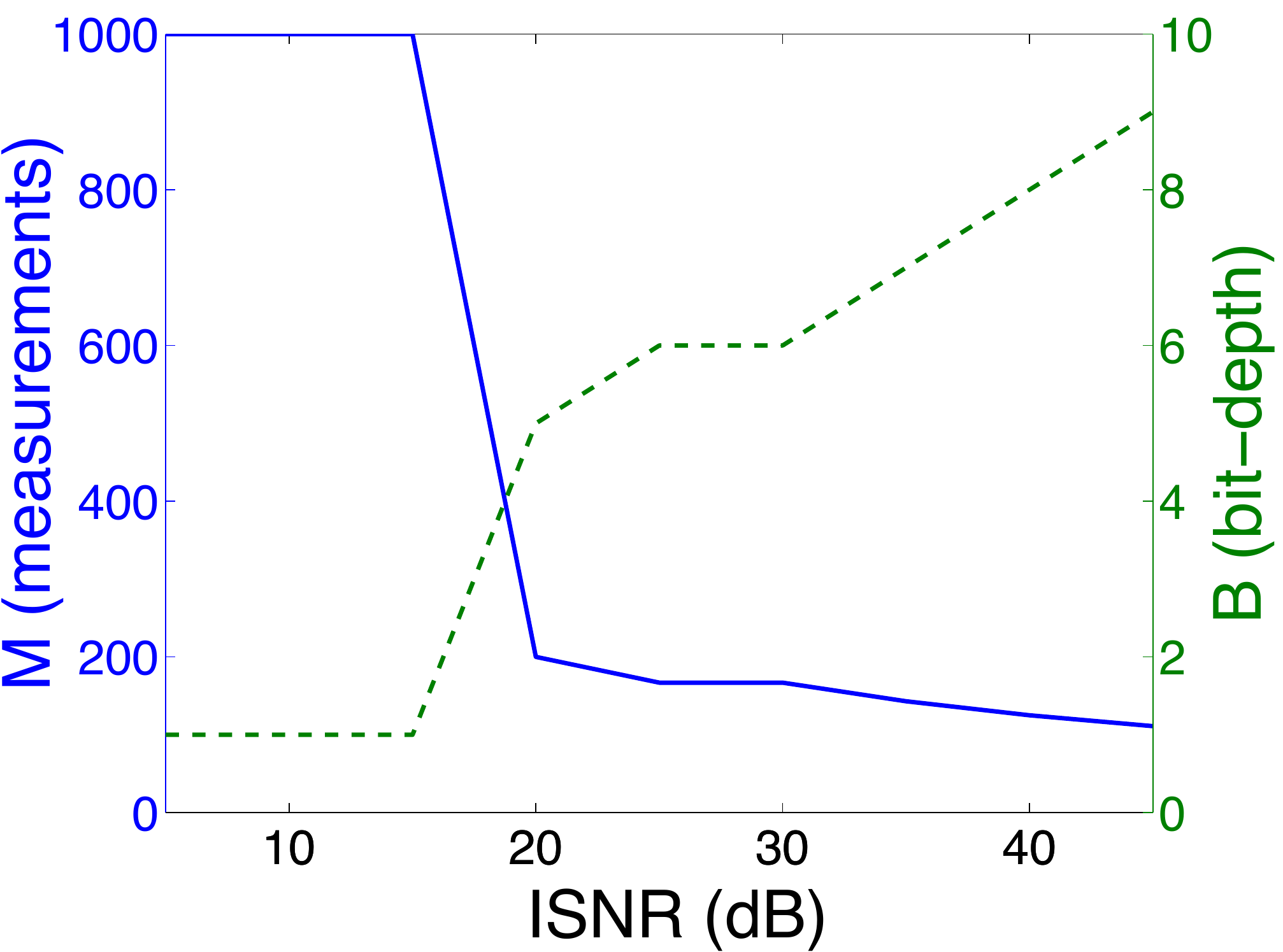}&
       \includegraphics[width=.55\imgwidth]{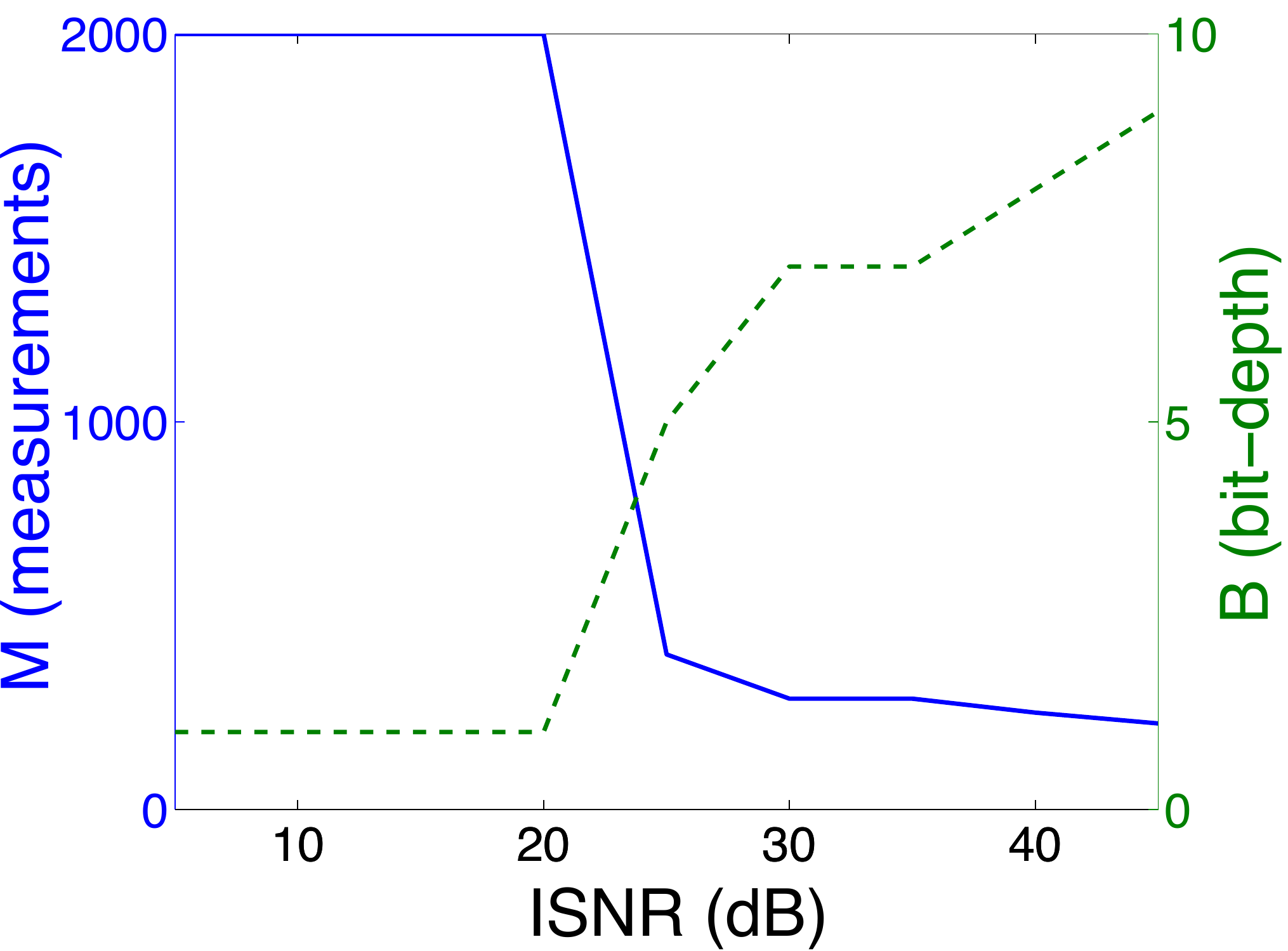}&
              \includegraphics[width=.55\imgwidth]{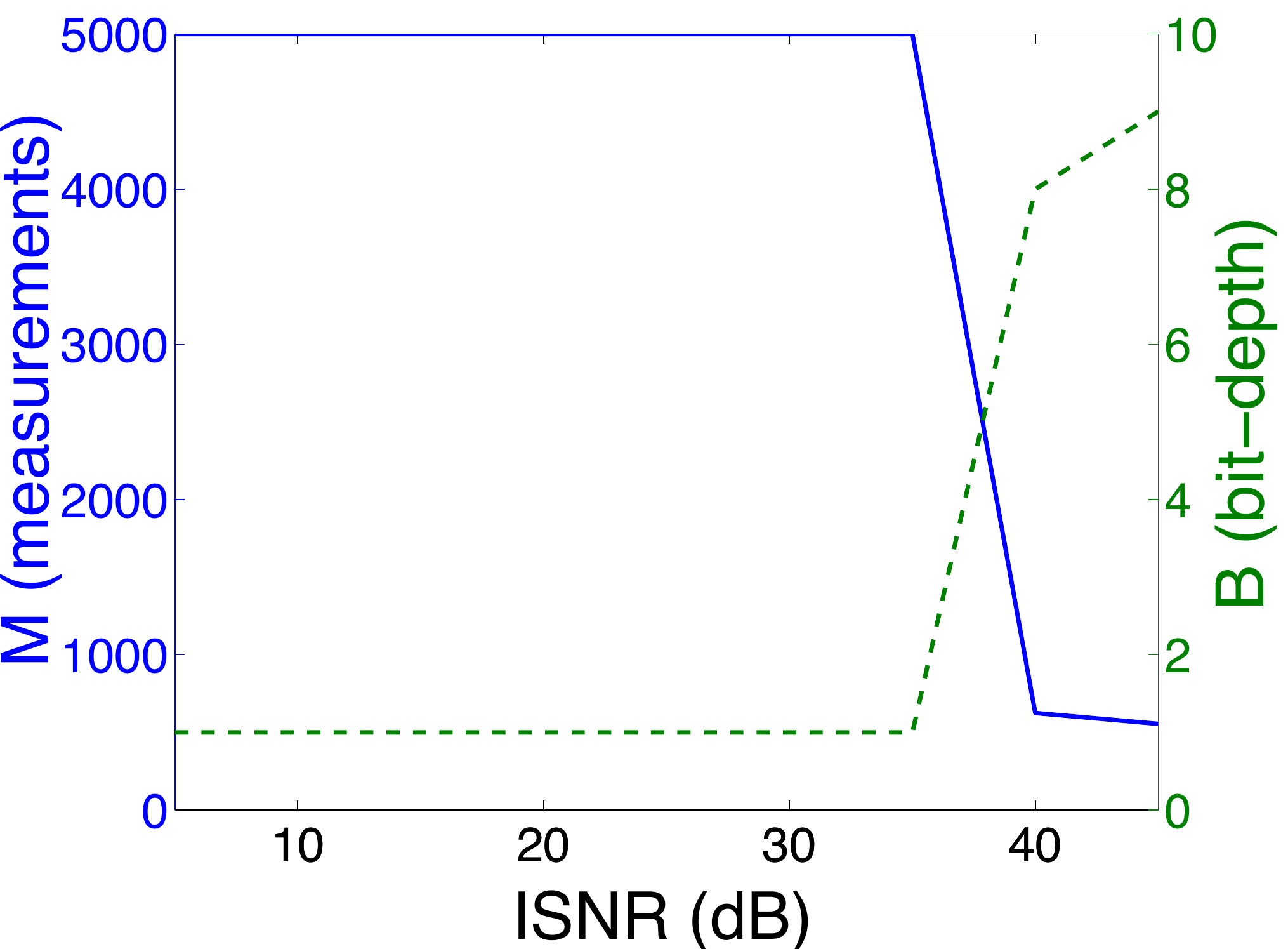}\\
       (a) $\mathfrak{B} = N$ & (b) $\mathfrak{B} = 2N$ & (c) $\mathfrak{B} = 5N$\\
   \end{tabular}
   \caption{Maximum RSNR given a fixed bit-budget $\mathfrak{B}$ for parameters $N=1000$, $K=10$.  The left side of each plot corresponds to the QC regime, while the right side corresponds to the MC regime. The solid line (blue) corresponds to the number of measurements $M$, while the dashed line (green) corresponds to the bit-depth $B$. }
   \label{fig:MvSNRvB}
\end{figure*}
In this set of experiments, we varied the $\mathrm{ISNR}$ between $5$dB and $45$dB and searched for the $(M,B)$ pair that maximized the RSNR, for a fixed bit-budget $\mathfrak{B}$ and parameters $N=1000$ and $K=10$.  As demonstrated by the previous experiment, the RSNR will not be the same for each bit-budget.

Figures~\ref{fig:MvSNRvB}(a)--(c) depict the results of this experiment for $\mathfrak{B} = N$, $2N$, and $5N$, respectively.  The left axis and solid line (blue) corresponds to the number of measurements $M$, while the right axis and dashed line (green) corresponds to the bit-depth $B$.  As always, we have that $\mathfrak{B} = MB$.  The QC regime is represented on the left side of the plots (low ISNR), while the MC regime is represented on the right side of the plots (high ISNR).  For example, for a bit-budget of $\mathfrak{B}=2N$, if the ISNR is $30$dB, then we are operating in the MC regime and should set the bit-depth to approximately $7$, resulting in the measurement ratio of approximately $M/N = 0.29$.  However, for the same bit-budget, if the ISNR is $15$dB, then we are operating in the QC regime and should set the bit-depth to $1$, resulting in a measurement ratio of $M/N = 2$.

In each plot in Figure~\ref{fig:MvSNRvB} there is a sharp transition between optimal bit-depth being high $(B \geq 5)$ and low $(B \leq 2)$.  This transition is centered at the $\mathrm{ISNR}$s $19$dB, $23$dB, and $38$dB, for the bit-budgets $\mathfrak{B} = N$, $2N$, and $5N$, respectively.  This implies that the transition occurs at higher ISNRs for higher bit-budgets.  Thus, we infer that, for higher bit-budgets $\mathfrak{B}$, it is better to choose low $B$, even when the input $\mathrm{ISNR}$ is fairly high.   The bottom line then is that, for moderate ISNR, the MC regime can be assumed when the bit-budget $\mathfrak{B}$ is small, while the QC regime can be assumed when the bit-budget is large.

\section{Discussion}
\label{sec:disc}

In this paper we have studied compressive sensing (CS) systems with scalar quantization when the total number of measurement bits is fixed and noise is present on the input signal.
Our results have demonstrated that \emph{in CS, it is sometimes better to reduce the bit-depth than the number of measurements}. We found that there exist two regimes: in the high-ISNR, MC regime, we should compress by reducing the number of measurements; this regime is best suited for applications where the acquisition if each measurement is expensive. In the low-ISNR, QC regime, we should compress by reducing the number of bits per measurement; this regime is best suited to applications where the acquisition of each measurement is cheap, or large bit-depth quantizers are expensive. 
The key to exposing the QC regime was the recognition that there is a tradeoff between amplified input signal noise (due to the underdetermined measurement system) and the number of bits that can be allocated per measurement under a fixed bit-budget.

Choosing a low bit-depth quantizer to reduce hardware complexity while driving up the sampling rate, as is recommended for the QC regime, is not a new idea.  Indeed, this same principle is the motivational force behind sigma-delta ADCs~\cite{bib:Candy92,bib:Aziz96,bib:Benedetto06,Bou::2006::Quantization-and-erasures} and other non-CS oversampled ADC architectures~\cite{CveDau::2007::Single-Bit-Oversampled,GoyVetTha::1998::Quantized-overcomplete,HoySadArc::2005::Monobit-digital}.  However, the ideas presented here differ significantly from previous oversampled ADC architectures in the following ways: \emph{i}) {CS is compressive}: Small bit-depth CS systems are expected to be used in cases where the bit-budget is significantly lower than in a conventional oversampled ADC system.  The use of sparse signal models enables \emph{compression}, i.e., a reduction in the total number of acquired bits, as opposed to just efficient sampling.  
\emph{ii}) {CS is non-adaptive}: As described earlier, CS measurement systems are non-adaptive, meaning they do not depend on the input signal.  This is true even for the $1$-bit CS case.  Almost all previous oversampled ADCs require some kind of feedback during quantization to produce stable representations. 
These differences place low bit-depth CS systems in a unique class of their own.  In a few words, CS, like physics has ``plenty of room at the bottom~\cite{feynmanRoomAtBottom}.''  

\appendix
\section{Proof of Theorem~\ref{thm:reconerrorbound}}
\label{apx:proof}

We first extend the upper bound of Theorem~4.1 in \cite{DavLasTre::2011::The-pros-and-cons} on the oracle-assisted reconstruction error to account for correlated measurement noise.
\begin{lemma}\label{lem:oracleCorr}
Suppose that $\bs y = \Phi \bs x + \bs z,$ where $\bs z \in \mathbb{R}^{M}$ is a zero-mean, random vector with covariance matrix $\Sigma = \mathbb{E}(\bs z \bs z^{T})$, and that $\bs x$ is $K$-sparse.  Furthermore, suppose that $\Phi$ satisfies the RIP of order $K$ with constant $\delta$.  Then the estimate $\widehat{\bs x}$ provided by the oracle-assisted reconstruction algorithm (\ref{eq:oraclerecon}) satisfies
\begin{equation}
\mathbb{E}\left( \| \bs x - \widehat{\bs x}  \|_{2}^{2}\right) \leq \frac{K}{1-\delta}\lambda_{\mathrm{max}}(\Sigma),
\end{equation}
where $\lambda_{\mathrm{max}}(\Sigma)$ is the largest eigenvalue of $\Sigma$.
\end{lemma}
\begin{proof}
For a fixed support set $\Omega \in \{ 1, \ldots, N\}$ with $|\Omega| = K$, the RIP ensures that $\Phi_{\Omega}$ is full rank, and thus the oracle estimate satisfies
\begin{equation}
\widehat{\bs x}|_{\Omega} = \bs x|_{\Omega} + \Phi^{\dagger}_{\Omega}\bs z.
\end{equation}
We seek to estimate $\mathbb{E}\left(\| \Phi_{\Omega}^{\dagger} \bs z \|_{2}^{2}\right)$.  

For any $K\times M$ matrix $A$ we have that
\begin{eqnarray}
\mathbb{E}\left( \| A \bs z \|_{2}^{2} \right) &=& \mathbb{E}(\mathrm{Tr}(A\bs z(A\bs z)^{T}))= \mathbb{E}(\mathrm{Tr}(A\bs z\bs z^{T} A^{T}))\nonumber\\
&=& \mathrm{Tr}(A\mathbb{E}(\bs z\bs z^{T}) A^{T}) = \mathrm{Tr}(A\Sigma A^{T})\nonumber\\
\label{eq:convertToEig}
&=& \sum_{j=1}^{K} \lambda_{j}(A\Sigma A^{T}),
\end{eqnarray}
where $\lambda_{j}(A\Sigma A^{T})$ denotes the $j$-th eigenvalue of $A\Sigma A^{T}$, and (\ref{eq:convertToEig}) follows since $A\Sigma A^{T}$ is a $K\times K$ matrix.
Lemma~8.2 of \cite{DavLasTre::2011::The-pros-and-cons} explains that the eigenvalues of this matrix can be upper bounded as 
\begin{eqnarray}
\lambda_{\mathrm{max}}(A\Sigma A^{T}) &\leq&\lambda_{\mathrm{max}}(AA^{T}) \lambda_{\mathrm{max}}(\Sigma)\nonumber \\
\label{eq:eigbound}
&\leq&   s_{\mathrm{max}}(A)^{2}\lambda_{\mathrm{max}}(\Sigma),
\end{eqnarray}
where $s_{\mathrm{max}}(A)$ denotes the maximum singular value of $A$.

Thus, to obtain the final bound, we combine (\ref{eq:convertToEig}) with (\ref{eq:eigbound}) and substitute $A = \Phi_{\Omega}^{\dagger}$, yielding
\begin{eqnarray}
\mathbb{E}\left(\| \Phi_{\Omega}^{\dagger} \bs z \|_{2}^{2} \right) &\leq& Ks_{\mathrm{max}}(\Phi_{\Omega}^{\dagger})^{2}\lambda_{\mathrm{max}}(\Sigma)\nonumber\\
&\leq& \frac{K}{1-\delta}\lambda_{\mathrm{max}}(\Sigma),
\end{eqnarray}
since we have that $s_{\mathrm{max}}(\Phi_{\Omega}^{\dagger})^{2} \leq \frac{1}{1-\delta}$ from Lemma~8.1 of \cite{DavLasTre::2011::The-pros-and-cons}.
\end{proof}

We next demonstrate that, by choosing  a signal model with random values and supports, the noiseless measurements $\Phi \bs x$ are identically distributed and uncorrelated.
\begin{lemma}\label{lem:sigmeas}
Let $\bs x \in \mathbb{R}^{N}$ be a sparse signal with support $\Omega \in \{1, \ldots, N\}$ and $|\Omega|=K$, where the elements $\Omega$ are chosen uniformly at random and the amplitudes of the non-zero coefficients are drawn according to $x_{j}\in \Omega \sim \mathcal{N}(0,\sigma_{\bs x}^{2})$.  Furthermore, let the $M\times N$ matrix $\Phi$ satisfy $\Phi\Phi^{T} = \frac{N}{M}\bs I_{M}$.  Then the vector $\Phi\bs x$ is distributed as a mixture of Gaussians with
\begin{equation}
\mathbb{E}((\Phi \bs x)_{i}) = 0,  \qquad  \mathbb{E}((\Phi \bs x)(\Phi \bs x)^{T}) = \frac{K}{M}\sigma_{\bs x}^{2}\bs I_{M},
\end{equation}
i.e., the elements $(\Phi\bs x)_{i}$ of $\Phi\bs x$ are zero-mean uncorrelated variables.
\end{lemma}
\begin{proof}

For a fixed support $\Omega$, each element $(\Phi\bs x)_{i}$ is Gaussian distributed with mean zero since it is the sum of $K$ zero-mean Gaussian variables.  Furthermore, the distribution of $(\Phi\bs x)_{i}$ over all possible supports is the sum of the distribution for each fixed support, scaled by the probability that they occur. Thus, $(\Phi\bs x)_{i}$ is a mixture of Gaussians with $\mathbb{E}( (\Phi\bs x)_{i}) = 0$.

To derive the variance of the elements and also show that they are uncorrelated, we first examine $\mathbb{E}(\bs x \bs x^{T})$.  The off-diagonal elements are zero, i.e., $\mathbb{E}(x_{i} x_{j})_{i\neq j} = 0$,
since the elements of $\bs x$ are uncorrelated, by definition.  Furthermore, the variance of the diagonal elements can be computed as
\begin{eqnarray}
\mathbb{E}(x_{i}^{2}) = \sigma_{\bs x}^{2}\mathbb{P}(i \in \Omega) =  \frac{K}{N}\sigma_{\bs x}^{2},\nonumber
\end{eqnarray}
since the $K$ non-zero support locations are chosen uniformly, any location $j$ is chosen with probability $K/N$.
Thus, $\mathbb{E}(\bs x \bs x^{T}) = \frac{K}{N}\sigma_{\bs x}^{2}\bs I_{N}$.  We next compute the correlation of the measurements $\Phi\bs x$ to obtain
\begin{eqnarray}
\mathbb{E}(\Phi\bs x (\Phi \bs x)^{T}) &=& \Phi \mathbb{E}(\bs x \bs x^{T}) \Phi^{T}\nonumber\\
&=&\frac{K}{N}\sigma_{\bs x}^{2}\Phi\Phi^{T} = \frac{K}{M}\sigma_{\bs x}^{2}\bs I_{M},
\end{eqnarray}
which concludes the proof.
\end{proof}

\begin{proof}[Proof of Theorem 1.]
Denote the error between the noiseless ideal measurements and $\bs y_{Q}$ by
\begin{equation}
\bs z := \Phi \bs x - \mathcal{Q}_{B}(\Phi \bs x + \Phi \bs n).
\end{equation}
Our goal is to determine a bound on the variance $\sigma_{z_{i}}^{2}$ of each element $z_{i}$ of $\bs z$.  We begin by rewriting the norm squared of $\bs z$ as
\begin{eqnarray}
z_{i}^{2} &=& [(\Phi \bs x)_{i} - \mathcal{Q}_{B}(\Phi \bs x + \Phi \bs n)_{i})]^{2}\nonumber\\
&=&  [(\Phi \bs x + \Phi \bs n)_{i} - \mathcal{Q}_{B}(\Phi \bs x + \Phi \bs n)_{i} - (\Phi \bs n)_{i} ]^{2}\nonumber \\
\label{eq:zibound}
&\leq&  2 [(\Phi \bs x + \Phi \bs n)_{i} - \mathcal{Q}_{B}(\Phi \bs x + \Phi \bs n)_{i}]^{2} + 2(\Phi \bs n)_{i}^{2},
\end{eqnarray}
where the index $i$ denotes individual elements of the respective vector.

We now seek an upper bound on the expected value of each of the quantities in (\ref{eq:zibound}). 
We begin with the second term in (\ref{eq:zibound}).  From the definition of $\Phi$, we have that the elements of $\Phi \bs n$ have variance 
\begin{equation}
\label{eq:noisefoldu}
\sigma_{\Phi \bs n}^{2} = \mathbb{E}((\Phi\bs n)_{i}^{2})  =  \frac{N}{M}\sigma_{\bs n}^{2},
\end{equation}
and furthermore are uncorrelated,
as was reviewed in Section~\ref{sec:background}. 

To bound the first term in  (\ref{eq:zibound}), we note that the optimal scalar
quantizer of rate $B$ for a Gaussian variable $g$ with
variance $\sigma^{2}$ has MSE given by $\mathbb{E}(g -
\mathcal{Q}_{B}(g))^{2} = \sigma^{2}2^{-2B}$.  
Furthermore, the MSE of an optimal quantizer of rate $B$ for any variable with variance
$\sigma^{2}$ is upper bounded by that of a Gaussian variable.
Our goal is to apply this quantization bound to $(\Phi \bs x + \Phi \bs n)_{i}$.  Since $(\Phi \bs x)_{i}$ and $(\Phi \bs n)_{i}$ are zero mean and independent of each other, then we immediately have that $\mathbb{E}\left((\Phi \bs x + \Phi \bs n)_{i}^{2}\right) = \frac{K}{M}\sigma_{\bs x}^{2} +  \frac{N}{M}\sigma_{\bs n}^{2}$,
where the first term follows from Lemma~\ref{lem:sigmeas}, and the second term follows from (\ref{eq:noisefoldu}). 
Thus, we can bound the first term in
(\ref{eq:zibound}) as
\begin{eqnarray}
\mathbb{E}\left( [(\Phi \bs x + \Phi \bs n)_{i} - \mathcal{Q}_{B}(\Phi \bs x + \Phi \bs n)_{i}]^{2}\right) &\leq& \mathbb{E}\left((\Phi \bs x + \Phi \bs n)_{i}^{2}\right)2^{-2B} \nonumber\\
\label{eq:firsttermbound}
&\leq& \frac{K}{M}\sigma_{\bs x}^{2}2^{-2B} + \frac{N}{M}\sigma_{\bs n}^{2}2^{-2B}.
\end{eqnarray}
Combining (\ref{eq:noisefoldu}) and (\ref{eq:firsttermbound}) as in (\ref{eq:zibound}) yields
\begin{equation}\label{eq:errvarfinal}
\sigma_{z_{i}}^{2} \leq 2\frac{K}{M}\sigma_{\bs x}^{2}2^{-2B}  + 2\frac{N}{M}\sigma_{\bs n}^{2}\left(1+2^{-2B}\right).
\end{equation}

We have thus far established an upper bound on the variance $\sigma_{z_{i}}^{2}$ of the error $z_{i}$ of each measurement.  We next obtain a bound on the eigenvalues of the covariance matrix $\Sigma = \mathbb{E}(\bs z \bs z^{T})$.  The off-diagonal elements of $\Sigma$ can be written as
\begin{eqnarray}
\mathbb{E}(z_{i}z_{j})_{i\neq j} &=& \mathbb{E}((\Phi\bs x)_{i}(\Phi\bs x)_{j})) - \mathbb{E}((\Phi\bs x)_{i}\mathcal{Q}_{B}(\Phi \bs x + \Phi \bs n)_{j}) \nonumber\\ && \qquad - \mathbb{E}((\Phi\bs x)_{j}\mathcal{Q}_{B}(\Phi \bs x + \Phi \bs n)_{i}) + \mathbb{E}(\mathcal{Q}_{B}(\Phi \bs x + \Phi \bs n)_{i}\mathcal{Q}_{B}(\Phi \bs x + \Phi \bs n)_{j}) \nonumber\\
&=& - \mathbb{E}(\mathcal{Q}_{B}(\Phi \bs x + \Phi \bs n)_{i}\mathcal{Q}_{B}(\Phi \bs x + \Phi \bs n)_{j}),
\end{eqnarray}
since $\mathbb{E}((\Phi\bs x)_{i}(\Phi\bs x)_{j}))=0$ by design and, for an optimal scalar quantizer, we have that $\mathbb{E}(\mathcal{Q}_{B}(\Phi \bs x + \Phi \bs n)_{i}\mathcal{Q}_{B}(\Phi \bs x + \Phi \bs n)_{j}) = \mathbb{E}((\Phi\bs x)_{j}\mathcal{Q}_{B}(\Phi \bs x + \Phi \bs n)_{i})$~\cite{GraNeu::1998::Quantization}.  Thus, the matrix $\Sigma$ has $\sigma_{z_{i}}^{2}$ along its diagonal and $\mathfrak{S}$ for all other entries.  We next apply Gershgorin's circle theorem,  which explains that any eigenvalue is upper bounded by the diagonal entry plus the sum of the magnitudes of the off-diagonal entries of each row of $\Sigma$.  Thus, we have 
\begin{equation}
\label{eq:gersh}
\lambda_{\mathrm{max}}(\Sigma) \leq \sigma_{z_{i}}^{2} + (M-1)\mathfrak{S},
\end{equation}
where $\mathfrak{S} = \max_{i\neq j}|\mathbb{E}(z_{i}z_{j})|$.

To obtain the final bound, we combine to (\ref{eq:errvarfinal}) with (\ref{eq:gersh}) and apply the upper bound in Lemma~\ref{lem:oracleCorr}.  We express the bound with the substitution $M = \mathfrak{B}/B$. 
\end{proof}


\section*{Acknowledgments}
Thanks to Petros Boufounos, Mark Davenport, Vivek Goyal, Laurent Jacques, Christoff Studer, and John Treichler for useful discussions and sage advice.

\bibliographystyle{IEEEtran}
\footnotesize
\bibliography{mai.bbl}

\begin{thebibliography}{10}
\providecommand{\url}[1]{#1}
\csname url@samestyle\endcsname
\providecommand{\newblock}{\relax}
\providecommand{\bibinfo}[2]{#2}
\providecommand{\BIBentrySTDinterwordspacing}{\spaceskip=0pt\relax}
\providecommand{\BIBentryALTinterwordstretchfactor}{4}
\providecommand{\BIBentryALTinterwordspacing}{\spaceskip=\fontdimen2\font plus
\BIBentryALTinterwordstretchfactor\fontdimen3\font minus
  \fontdimen4\font\relax}
\providecommand{\BIBforeignlanguage}[2]{{%
\expandafter\ifx\csname l@#1\endcsname\relax
\typeout{** WARNING: IEEEtran.bst: No hyphenation pattern has been}%
\typeout{** loaded for the language `#1'. Using the pattern for}%
\typeout{** the default language instead.}%
\else
\language=\csname l@#1\endcsname
\fi
#2}}
\providecommand{\BIBdecl}{\relax}
\BIBdecl

\bibitem{Can::2006::Compressive-sampling}
E.~Cand\`{e}s, ``Compressive sampling,'' in \emph{Proc. Int. Congress Math.},
  Madrid, Spain, Aug. 2006.

\bibitem{Don::2006::Compressed-sensing}
D.~Donoho, ``Compressed sensing,'' \emph{IEEE Trans. Inform. Theory}, vol.~6,
  no.~4, pp. 1289--1306, 2006.

\bibitem{CandesDLP}
E.~Cand\`{e}s and T.~Tao, ``Decoding by linear programming,'' \emph{IEEE Trans.
  Inform. Theory}, vol.~51, no.~12, pp. 4203--4215, 2005.

\bibitem{TroppG_Signal}
J.~Tropp and A.~Gilbert, ``Signal recovery from partial information via
  orthogonal matching pursuit,'' \emph{IEEE Trans. Inform. Theory}, vol.~53,
  no.~12, pp. 4655--4666, 2007.

\bibitem{TroLasDua::2009::Beyond-Nyquist:}
J.~Tropp, J.~Laska, M.~Duarte, J.~Romberg, and R.~Baraniuk, ``Beyond {N}yquist:
  {E}fficient sampling of sparse, bandlimited signals,'' \emph{IEEE Trans.
  Inform. Theory}, vol.~56, no.~1, pp. 520--544, 2010.

\bibitem{DuaDavTak::2008::Single-pixel-imaging}
M.~Duarte, M.~Davenport, D.~Takhar, J.~Laska, T.~Sun, K.~Kelly, and
  R.~Baraniuk, ``Single-pixel imaging via compressive sampling,'' \emph{IEEE
  Signal Processing Mag.}, vol.~25, no.~2, pp. 83--91, 2008.

\bibitem{SlaviLDB_Compressive}
J.~P. Slavinsky, J.~Laska, M.~Davenport, and R.~Baraniuk, ``The compressive
  mutliplexer for multi-channel compressive sensing,'' in \emph{IEEE Int. Conf.
  Acoustics, Speech, and Signal Processing ({ICASSP})}, Prague, Czech Republic,
  May 2011.

\bibitem{BajwaHRWN_Toeplitz}
W.~Bajwa, J.~Haupt, G.~Raz, S.~Wright, and R.~Nowak, ``Toeplitz-structured
  compressed sensing matrices,'' in \emph{Proc. 14th IEEE/SP Workshop
  Statistical Signal Processing (SSP'07)}, Madison, WI, Aug. 2007.

\bibitem{Duarte_spectral}
M.~Duarte and R.~Baraniuk, ``Spectral compressive sensing,'' 2010, {P}reprint.

\bibitem{HegdeDC_Compressive}
C.~Hegde, M.~Duarte, and V.~Cevher, ``Compressive sensing recovery of spike
  trains using a structured sparsity model,'' in \emph{Signal Processing with
  Adaptive Sparse Structured Representations (SPARS)}, Saint-Malo, France, Apr.
  2009.

\bibitem{BaranCDH_Model}
R.~Baraniuk, V.~Cevher, M.~Duarte, and C.~Hegde, ``Model-based compressive
  sensing,'' \emph{IEEE Trans. Inform. Theory}, vol.~56, no.~4, pp. 1982--2001,
  2010.

\bibitem{HaleYZ_Fixed}
E.~Hale, W.~Yin, and Y.~Zhang, ``A fixed-point continuation method for
  $\ell_1$-regularized minimization with applications to compressed sensing,''
  Rice Univ., CAAM Dept., Tech. Rep. TR07-07, 2007.

\bibitem{YinOGD_Bregman}
W.~Yin, S.~Osher, D.~Goldfarb, and J.~Darbon, ``Bregman iterative algorithms
  for $\ell_1$-minimization with applications to compressed sensing,''
  \emph{SIAM J. Imag. Sci.}, vol.~1, no.~1, pp. 143--168, 2008.

\bibitem{FigueNW_Gradient}
M.~Figueiredo, R.~Nowak, and S.~Wright, ``Gradient projections for sparse
  reconstruction: Application to compressed sensing and other inverse
  problems,'' \emph{IEEE J. Select. Top. Signal Processing}, vol.~1, no.~4, pp.
  586--597, 2007.

\bibitem{BergF_Probing}
E.~{van den Berg} and M.~Friedlander, ``Probing the {P}areto frontier for basis
  pursuit solutions,'' \emph{SIAM J. on Sci. Comp.}, vol.~31, no.~2, pp.
  890--912, 2008.

\bibitem{cosamp}
D.~Needell and J.~Tropp, ``Co{S}a{MP}: Iterative signal recovery from
  incomplete and inaccurate samples,'' \emph{Appl. Comput. Harmon. Anal.},
  vol.~26, no.~3, pp. 301--321, 2009.

\bibitem{BluDav::2008::Iterative-hard}
T.~Blumensath and M.~Davies, ``Iterative hard thresholding for compressive
  sensing,'' \emph{Appl. Comput. Harmon. Anal.}, vol.~27, no.~3, pp. 265--274,
  2009.

\bibitem{DonohMM_Message}
D.~Donoho, A.~Maleki, and A.~Montanari, ``Message passing algorithms for
  compressed sensing,'' \emph{Proc. Natl. Acad. Sci.}, vol. 106, no.~45, pp.
  18\,914--18\,919, 2009.

\bibitem{LustiDP_Rapid}
M.~Lustig, D.~Donoho, and J.~Pauly, ``Rapid {MR} imaging with compressed
  sensing and randomly under-sampled {3DFT} trajectories,'' in \emph{Proc.
  Annual Meeting of ISMRM}, Seattle, WA, May 2006.

\bibitem{Hea05:Analog-to-Information}
\BIBentryALTinterwordspacing
D.~Healy. (2005) Analog-to-information. DARPA BAA \#05-35. [Online]. Available:
  \url{http://www.darpa.mil/mto/solicitations/baa05-35/s/index.html}
\BIBentrySTDinterwordspacing

\bibitem{GunPowSaa::2010::Sobolev-Duals}
C.~S. Gunturk, A.~Powell, R.~Saab, and O.~Yilmaz, ``Sobolev duals for random
  frames and sigma-delta quantization of compressed sensing measurements,''
  2010, preprint.

\bibitem{Boufounos::2010::univer_rate_effic_scalar_quant}
P.~Boufounos, ``Universal rate-efficient scalar quantization,'' 2010, preprint.

\bibitem{bib:GPSY_SD10}
C.~S. Gunturk, A.~Powell, R.~Saab, and O.~Yilmaz, ``{Sobolev Duals for Random
  Frames and Sigma-Delta Quantization of Compressed Sensing Measurements},''
  2010, preprint.

\bibitem{ZymBoyCan::2009::Compressed-sensing}
A.~Zymnis, S.~Boyd, and E.~Cand\`{e}s, ``Compressed sensing with quantized
  measurements,'' \emph{IEEE Sig. Proc. Letters}, vol.~17, no.~2, Feb. 2010.

\bibitem{SunGoy::2009::Quantization-for-compressed}
J.~Sun and V.~Goyal, ``Quantization for compressed sensing reconstruction,'' in
  \emph{Proc. Sampling Theory and Applications (SampTA)}, Marseille, France,
  May 2009.

\bibitem{SunGoy::20090::Optimal-quantization}
J.~Z. Sun and V.~K. Goyal, ``Optimal quantization of random measurements in
  compressed sensing,'' in \emph{Int. Sym. on Inform. Theory (ISIT)}, June
  2009.

\bibitem{vivekQuantFrame}
H.~Q. Nguyen, V.~Goyal, and L.~Varshney, ``Frame permutation quantization,''
  \emph{Appl. Comput. Harmon. Anal.}, Nov. 2010.

\bibitem{LasBouDav::2009::Demcracy-in-action:}
J.~Laska, P.~Boufounos, M.~Davenport, and R.~Baraniuk, ``Democracy in action:
  {Q}uantization, saturation, and compressive sensing,'' \emph{to appear in
  App. Comp. and Harm. Anal.}, 2011.

\bibitem{JacquHF_Dequantizing}
L.~Jacques, D.~Hammond, and M.~Fadili, ``Dequantizing compressed sensing:
  {W}hen oversampling and non-gaussian contraints combine,'' \emph{IEEE Trans.
  Inform. Theory}, vol.~57, no.~1, pp. 559--571, 2009.

\bibitem{CanDav::2011::How-well-can-we-estimate}
E.~Cand{\`e}s and M.~A. Davenport, ``How well can we estimate a sparse
  vector?'' 2011, preprint.

\bibitem{DASP}
J.~Treichler, M.~Davenport, and R.~Baraniuk, ``Application of compressive
  sensing to the design of wideband signal acquisition receivers,'' in
  \emph{U.S./Australia Joint Work. Defense Apps. of Signal Processing (DASP)},
  Lihue, Hawaii, Sept. 2009.

\bibitem{DavLasTre::2011::The-pros-and-cons}
M.~Davenport, J.~Laska, J.~Treichler, and R.~Baraniuk, ``The pros and cons of
  compressive sensing: Noise folding and dynamic range,'' 2011, preprint.

\bibitem{CoverT_Elements}
T.~Cover and J.~Thomas, \emph{Elements of Information Theory}.\hskip 1em plus
  0.5em minus 0.4em\relax New York, NY: Wiley-Interscience, 1991.

\bibitem{SarvoBB_Measurements}
S.~Sarvotham, D.~Baron, and R.~Baraniuk, ``Measurements vs. bits: Compressed
  sensing meets information theory,'' in \emph{Proc. Allerton Conf.
  Communication, Control, and Computing}, Monticello, IL, Sept. 2006.

\bibitem{JacLasBou::2011::Robust-1-bit}
L.~Jacques, J.~Laska, P.~Boufounos, and R.~Baraniuk, ``Robust 1-bit compressive
  sensing via binary $\epsilon$-stable embeddings,'' 2011, preprint.

\bibitem{LeRonRee::2005::Analog-to-Digital-Converters}
B.~Le, T.~W. Rondeau, J.~H. Reed, and C.~W. Bostian, ``Analog-to-digital
  converters,'' \emph{IEEE Sig. Proc. Mag.}, Nov. 2005.

\bibitem{BouBar::2008::1-Bit-compressive}
P.~Boufounos and R.~Baraniuk, ``1-bit compressive sensing,'' in \emph{Proc.
  Conf. Inform. Science and Systems (CISS)}, Princeton, NJ, Mar. 2008.

\bibitem{LasWenYin::2010::Trust-but-verify:}
J.~Laska, Z.~Wen, W.~Yin, and R.~Baraniuk, ``Trust, but verify: Fast and
  accurate signal recovery from 1-bit compressive measurements,'' \emph{IEEE
  Trans. Sig. Proc.}, vol.~59, no.~11, pp. 5289--5301, 2011.

\bibitem{Bou::2009::Greedy-sparse}
P.~Boufounos, ``Greedy sparse signal reconstruction from sign measurements,''
  in \emph{Proc. Asilomar Conf. on Signals Systems and Comput.}, Asilomar,
  California, Nov. 2009.

\bibitem{CandeT_Dantzig}
E.~Cand\`{e}s and T.~Tao, ``The {D}antzig selector: {S}tatistical estimation
  when $p$ is much larger than $n$,'' \emph{Ann. Stat.}, vol.~35, no.~6, pp.
  2313--2351, 2007.

\bibitem{PlaVer::2011::One-bit-compressed}
Y.~Plan and R.~Vershynin, ``One-bit compressed sensing by linear programming,''
  2011, preprint.

\bibitem{CandeT_Decoding}
E.~Cand\`{e}s and T.~Tao, ``Decoding by linear programming,'' \emph{IEEE Trans.
  Inform. Theory}, vol.~51, no.~12, pp. 4203--4215, 2005.

\bibitem{GraNeu::1998::Quantization}
R.~M. Gray and D.~L. Neuhoff, ``Quantization,'' \emph{IEEE Trans. Info.
  Theory}, vol.~44, no.~6, pp. 2325--2383, 1998.

\bibitem{bib:Candy92}
J.~C. Candy and G.~C. Temes, Eds., \emph{Oversampling {D}elta-{S}igma
  Converters}.\hskip 1em plus 0.5em minus 0.4em\relax IEEE Press, 1992.

\bibitem{bib:Aziz96}
P.~M. Aziz, H.~V. Sorensen, and J.~V.~D. Spiegel, ``An overview of
  {S}igma-{D}elta converters,'' \emph{IEEE Sig. Proc. Mag.}, vol.~13, no.~1,
  pp. 61--84, Jan. 1996.

\bibitem{bib:Benedetto06}
J.~J. Benedetto, A.~M. Powell, and O.~Yilmaz, ``Sigma-delta quantization and
  finite frames,'' \emph{IEEE Trans. Info. Theory}, vol.~52, no.~5, pp.
  1990--2005, May 2006.

\bibitem{Bou::2006::Quantization-and-erasures}
P.~Boufounos, ``{Quantization and Erasures in Frame Representations},'' Ph.D.
  dissertation, MIT EECS, Cambridge, MA, Jan. 2006.

\bibitem{CveDau::2007::Single-Bit-Oversampled}
Z.~Cvetkovi\'{c} and I.~Daubechies, ``Single-bit oversampled {A}/{D} conversion
  with exponential accurace in the bit-rate,'' \emph{IEEE Trans. Info. Theory},
  vol.~53, no.~11, pp. 3979--3989, 2007.

\bibitem{GoyVetTha::1998::Quantized-overcomplete}
V.~K. Goyal, M.~Vetterli, and N.~Thao, ``Quantized overcomplete expansions in
  $\mathbb{R}^n$: Analysis, synthesis, and algorithms,'' \emph{IEEE Trans.
  Inform. Theory}, vol.~44, no.~1, pp. 16--31, 1998.

\bibitem{HoySadArc::2005::Monobit-digital}
S.~Hoyos, B.~Sadler, and G.~Arce, ``Monobit digital receivers for ultrawideband
  communications,'' \emph{IEEE Trans. Wireless Comm.}, vol.~4, no.~4, pp.
  1337--1344, July 2005.

\bibitem{feynmanRoomAtBottom}
R.~Feynman, ``There's plent of room at the bottom,'' in \emph{Caltech Eng. and
  Sci.}, {American Physical Society}, Ed., vol.~23, no.~5, 1960, pp. 22--36.

\end{thebibliography}

\end{document}